\documentclass{bland}                     
\smartqed  
\usepackage{graphicx}
\usepackage{amsmath}
\usepackage{psfrag}

\def\SqOrig{A}
\def\SqLoch{\tilde{A}}
\def\SqLochklein{\tilde{a}}

\newcommand{\NP}{NP}

\newcommand{\PSPACE}{PSPACE}

\begin{document}

\title{Online Square Packing
}

\titlerunning{Online Square Packing}        

\author{S\'andor P. Fekete
\and Tom Kamphans\thanks{Tom Kamhans was supported by DFG grant 
FE 407/8-3, project ``ReCoNodes''. A preliminary extended abstract
summarizing the results of this paper appeared in ~\cite{fks-olsp-09p}.} 
\and Nils Schweer}

\authorrunning{S.P.\ Fekete, T.\ Kamphans, and N.\ Schweer} 

\institute{
Braunschweig University of Technology\\
Department of Computer Science, Algorithms Group\\
M\"uhlenpfordtstrasse 23, 38106 Braunschweig, Germany\\
http:www.ibr.cs.tu-bs.de/alg\\
\email{\{s.fekete, t.kamphans, n.schweer\}@tu-bs.de}.
}

%

\date{}

\maketitle

\begin{abstract}
We analyze the problem of packing squares in an online fashion:
Given a semi-infinite strip of width~1 and an unknown sequence of
squares of side length in $[0,1]$ that arrive from above, one at a
time. The objective is to pack these items as they arrive,
minimizing the resulting height. Just like in the classical game of
Tetris, each square must be moved along a collision-free path to its
final destination. In addition, we account for gravity in both
motion (squares must never move up) and position (any final
destination must be supported from below). A similar problem has
been considered before; the best previous result is by Azar and
Epstein, who gave a 4-competitive algorithm in a setting without
gravity (i.e., with the possibility of letting squares ``hang in the
air'') based on ideas of shelf-packing: Squares are assigned to
different horizontal levels, allowing an analysis that is
reminiscent of some bin-packing arguments. We apply a geometric
analysis to establish a competitive factor of 3.5 for the
bottom-left heuristic and present a $\frac{34}{13} \approx
2.6154$-competitive algorithm.


\end{abstract}

\section{Introduction}
\label{intro} 

\subsection{Packing Problems}
Packing problems arise in many different situations, either concrete (where
actual physical objects have to be packed), or abstract (where the space is
virtual, e.g., in scheduling). Even in a one-dimensional setting, computing
an optimal set of positions in a container for a known set of objects is a
classical, hard problem. Having to deal with two-dimensional objects adds a
variety of difficulties; one of them is the more complex structure of feasible
placements; see, for example, Fekete et al.~\cite{fsv-eahdop-07}. Another one is actually
moving the objects into their final locations without causing collisions or
overlap along the way.  A different kind of difficulty may arise from a lack of
information: in many settings, objects have to be assigned to their final
locations one by one, without knowing future items. Obviously, this makes the
challenge even harder.  

In this paper, we consider online packing of squares
into a vertical strip of unit width. Squares arrive from above in an online
fashion, one at a time, and have to be moved to their final positions. On this
path, a square may move only through unoccupied space, come to a stop only if
it is supported from below; in allusion to the well-known computer game, this
is called the {\em Tetris constraint}. In addition, an item is not allowed to move
upwards and has to be supported from below when reaching its final position;
these conditions are called {\em gravity constraints}. The objective is to minimize
the total height of the occupied part of the strip. 

\subsection{Problem Statement}
Let $S$ be a semi-infinite strip of width~1 and ${\cal
A}=(\SqOrig_1, \ldots, \SqOrig_n)$ a sequence of squares with side
length $a_i \leq 1$, $i=1,\ldots, n$. The sequence is unknown in ad-
\linebreak vance. A strategy gets the squares one by one and must
place a square before it gets the next. Initially, a square is
located above all previously placed ones.

Our goal is to find a non-overlapping packing of squares in the
strip that keeps the height of the occupied area as low as possible.
More precisely, we want to minimize the distance between the bottom
side of $S$ and the highest point that is occupied by a square. The
sides of the squares in the packing must be parallel to the sides of
the strip. Moreover, a packing must fulfill two additional
constraints:

{\em Tetris constraint}:\index{Tetris constraint} At the time a
square is placed, there is a collision-free path from the initial
position of a square (top of the strip) to the square's final
position.

{\em Gravity constraint}:\index{gravity constraint} A square must be
packed on top of another square
({\em i.e.,} the intersection of the upper square's bottom side and
the lower square's top side must be a line segment) or on the bottom
of the strip; in addition, no square may ever move up on the path to
its final position.

\subsection{Related Work}\label{sec:relatedwork_tetris}
\index{strip packing problem}

In general, a packing problem is defined by a set of items that have
to be packed into a container (a set of containers) such that some
objective function, {\em e.g.,} the area where no item is placed or
the number of used containers, is minimized. A huge amount of work
has been done on different kinds of packing problems. A survey on
approximation algorithms for packing problems can be found
in~\cite{s-capsp-02}.

A special kind of packing problem is the {\em strip packing
problem}. It asks for a non-overlapping placement of a set of
rectangles in a semi-infinite strip such that the height of the
occupied area is minimized. The bottom side of a rectangle has to be
parallel to the bottom side of the strip. Over the years, many
different variations of the strip packing problem have been
proposed: online, offline, with or without rotation, and so on.
Typical measures for the evaluation of approximation and online
algorithms are the absolute performance and the asymptotic
performance ratio.

If we restrict all rectangles to be of the same height, the strip
packing problem without rotation is equivalent to the {\em bin
packing problem}:\index{bin packing problem} Given a set of
one-dimensional items each having a size between zero and one, the
task is to pack these items into a minimum number of unit size bins.
Hence, all negative results for the bin packing problem, {\em e.g.,}
\NP-hardness and lower bounds on the competitive ratio also hold for
the strip packing problem; see~\cite{gw-obprs-95} for a survey on
(online) bin packing.

If we restrict all rectangles to be of the same width then the strip
packing problem without rotation is equivalent to the {\em list
scheduling problem}\index{list scheduling}:
Given a set of jobs with different processing times, the task is to
schedule these jobs on a set of identical machines such that the
makespan is minimized. This problem was first studied by
Graham~\cite{g-bmta-69}. There are many different kinds of
scheduling problems, {\em e.g.,} the machines can be identical or
not, {\em preemption} might be allowed or not, and there might be
other restrictions such as {\em precedence constraints} or {\em
release times}; see~\cite{b-sa-04} for a textbook on scheduling.

\paragraph{Offline Strip Packing}\index{strip packing problem!offline}
Concerning the absolute approximation factor, Baker {\em et
al.}~\cite{bcr-optd-80} introduce the class of {\em bottom-up
left-justified} algorithms. A specification that sorts the items in
advance is a $3$-approximation for a sequence of rectangles and a
$2$-approximation for a sequence of
squares. 
Sleator~\cite{s-toapt-80} presents an algorithm with approximation
factor $2.5$, Schiermeyer~\cite{s-rfoap-94} and
Steinberg~\cite{s-spaap-97} present algorithms that achieve an
absolute approximation factor of $2$, for a sequence of rectangles.

Concerning the asymptotic approximation factor, the algorithms
\linebreak presented by Coffman {\em et al.}~\cite{cgjt-pblot-80}
achieve performance bounds of $2$, $1.7$, and $1.5$. Baker {\em et
al.}~\cite{bbk-atdp-81} improve this factor to $1.25$. Kenyon and
R{\'e}mila~\cite{kr-asp-96} design a fully polynomial time
approximation scheme. Han {\em et al.}~\cite{hiyz-spvbp-07} show
that every algorithm for the bin packing problem implies an
algorithm for the strip packing problem with the same approximation
factor. Thus, in the offline case, not only the negative results but
also the positive results from bin packing hold for strip packing.

\paragraph{Online Strip Packing}\index{strip packing problem!online}
Concerning the absolute competitive ratio Baker {\em
et~al.}~\cite{bs-satdp-83} present two algorithms with competitive
ratio $7.46$ and $6.99$. If the input sequence consists only of
squares the competitive ratio reduces to $5.83$ for both algorithms.
These algorithms are the first {\em shelf algorithms:} A shelf
algorithm classifies the rectangles according to their height, {\em
i.e.,} a rectangle is in a class $s$ if its height is in the
interval $(\alpha^{s-1},\alpha^{s}]$, for a parameter $\alpha \in
(0,1)$. Each class is packed in a separate {\em shelf}, {\em i.e.,}
into a rectangular area of width one and height $\alpha^s$, inside
the strip. A bin packing algorithm is used as a subroutine to pack
the items. Ye {\em et al.}~\cite{yhz-nosp-09} present an algorithm
with absolute competitive factor $6.6623$. Lower bounds for the
absolute performance ratio are $2$ for sequences of rectangles and
$1.75$ for sequences of squares~\cite{bbk-lbotd-82}.

Concerning the asymptotic competitive ratio, the algorithms
in~\cite{bs-satdp-83} achieve a competitive ratio of $2$ and $1.7$.
Csirik and Woeginger~\cite{cw-saosp-97} show a lower bound of
$1.69103$ for any shelf algorithm and introduce a shelf algorithm
whose competitive ratio comes arbitrarily close to this value. Han
{\em et al.}~\cite{hiyz-spvbp-07} show that for the so called {\em
Super Harmonic} algorithms, for the bin packing problem, the
competitive ratio can be transferred to the strip packing problem.
The current best algorithm for bin packing is
$1.58889$-competitive~\cite{s-obpp-02}. Thus, there is an algorithm
with the same ratio for the strip packing problem. A lower bound,
due to van Vliet~\cite{v-ilbob-92}, for the asymptotic competitive
ratio, is $1.5401$. This bound also holds for sequences consisting
only of squares.

\paragraph{Tetris}\index{Tetris} Every reader is certainly familiar with the
classical game of Tetris: Given a strip of fixed width, find an
online placement for a sequence of objects falling down from above
such that space is utilized as good as possible. In comparison to
the strip packing problem, there is a slight difference in the
objective function as Tetris aims at filling rows. In actual
optimization scenarios this is less interesting as it is not
critical whether a row is used to precisely 100\%---in particular,
as full rows do not magically disappear in real life. In this
process, no item can ever move upward, no collisions between objects
must occur, an item will come to a stop if and only if it is
supported from below, and each placement has to be fixed before the
next item arrives. Even when disregarding the difficulty of
ever-increasing speed, Tetris is notoriously difficult: Breukelaar
{\em et al.}~\cite{bdhhkl-tihea-04} show that Tetris is
\PSPACE-hard, even for the, original, limited set of different
objects.

\paragraph{Strip Packing with Tetris Constraint}\index{strip packing problem!Tetris constraint}
Tetris-like online packing has been considered before. Most notably,
Azar and Epstein~\cite{ae-tdp-97} consider online packing of
rectangles into a strip; just like in Tetris, they consider the
situation with or without rotation of objects. For the case without
rotation, they show that no constant competitive ratio is possible,
unless there is a fixed-size lower bound of $\varepsilon$ on the
side length of the objects, in which case there is an upper bound of
$O(\log\frac{1}{\varepsilon})$ on the competitive ratio.

For the case in which rotation is possible, they present a
4-competitive strategy based on shelf-packing methods: Each
rectangle is rotated such that its narrow side is the bottom side.
The algorithm tries to maintain a corridor at the right side of the
strip to move the rectangles to their shelves. If a shelf is full or
the path to it is blocked, by a large item, a new shelf is opened.
Until now, this is also the best deterministic upper bound for
squares. Note that in this strategy gravity is not taken into
account as items are allowed to be placed
at appropriate levels.   

Coffman {\em et al.}~\cite{cdw-prs-02} consider probabilistic
aspects of online rectangle packing without rotation and with Tetris
constraint. If $n$ rectangle side lengths are chosen uniformly at
random from the interval $[0,1]$, they show that there is a lower
bound of $(0.31382733...)n$ on the expected height for any
algorithm. Moreover, they propose an algorithm that achieves an
asymptotic expected height of $(0.36976421...)n$.

\paragraph{Strip Packing with Tetris and Gravity Constraint}
\index{strip packing problem!Tetris and gravity constraint}
There is one negative result for the setting with Tetris and gravity
constraint when rotation is not allowed in~\cite{ae-tdp-97}: If all
rectangles have a width of at least $\varepsilon > 0$ or of at most
$1 - \varepsilon$, then the competitive factor of any algorithms is
$\Omega(\frac{1}{\varepsilon})$.

\subsection{Our Results} 
We analyze a natural and simple heuristic called {\em BottomLeft}
(Section~\ref{sect:bottomleft}), which works similar to the one
introduced by Baker {\em et al.}~\cite{bcr-optd-80}. We show that it
is possible to give a better competitive ratio than the ratio 4
achieved by Azar and Epstein, even in the presence of gravity. We
obtain an asymptotic competitive ratio of 3.5 for {\em BottomLeft}.
Furthermore, we introduce the strategy {\em SlotAlgorithm}
(Section~\ref{sect:stratSlotAlg}), which improves the upper bound to
$\frac{34}{13}=2.6154...$, asymptotically.

\section{The Strategy {\em BottomLeft}\label{sect:bottomleft}}
\index{BottomLeft}

In this section, we analyze the packing generated by the strategy
{\em BottomLeft}, which works as follows: We place the current
square as close as possible to the bottom of the strip; this means
that we move the square along a collision-free path from the top of
the strip to the desired position, without ever moving the square in
positive $y$-direction. We break ties by choosing the leftmost among
all possible bottommost positions.

\smallskip

A packing may leave areas of the strip empty. We call a maximal
connected component (of finite size) of the strip's empty area a
{\em hole},\index{hole} denoted by $H_h$, $h\in$
denote by $|H_h|$ the size of $H_h$. For a simplified analysis, we
finish the packing with an additional square, $\SqOrig_{n+1}$, of
side length 1. As a result, all holes have a closed boundary. Let
$H_1,\ldots, H_s$ be the holes in the packing. We can express the
height of the packing produced by BottomLeft as follows:
\begin{equation*}\label{eq:BLheight}
BL = \sum_{i=1}^n a_{i}^2 + \sum_{h=1}^s |H_h| \,.
\end{equation*}
In the following sections, we prove that
\begin{equation*}
\sum_{h=1}^s |H_h| \leq 2.5 \cdot \sum_{i=1}^{n+1} a_{i}^2 \, .
\end{equation*}
Because any strategy produces at least a height of $\sum_{i=1}^n
a_{i}^2$, and because $a_{n+1}^2 = 1$, we get
\begin{equation*}
BL = \sum_{i=1}^n a_{i}^2 + \sum_{h=1}^s |H_h| \leq \sum_{i=1}^n
a_{i}^2 + 2.5\cdot \sum_{i=1}^{n+1} a_{i}^2 \leq 3.5\cdot OPT + 2.5
\, ,
\end{equation*}
where $OPT$ denotes the height of an optimal packing.
%
This proves:
\begin{theorem}
BottomLeft is (asymptotically) $3.5$-competitive.
\end{theorem}


\paragraph{Definitions}
Before we start with the analysis, we need some definitions: We
denote the bottom (left, right) side of the strip by $B_S$ ($R_S$,
$L_S$; respectively), and the sides of a square, $\SqOrig_i$, by
$B_{\SqOrig_i}$, $T_{\SqOrig_i}$, $R_{\SqOrig_i}$, $L_{\SqOrig_i}$
(bottom, top, right, left; respectively); see
Fig.~\ref{fig:defs_psfrag}. The $x$-coordinates of the left and
right side of $\SqOrig_i$ in a packing are $l_{\SqOrig_i}$ and
$r_{\SqOrig_i}$; the $y$-coordinates of the top and bottom side
$t_{\SqOrig_i}$ and $b_{\SqOrig_i}$, respectively. Let the {\em left
neighborhood}, $N_L(\SqOrig_i)$, be the set of squares that touch
the left side of $\SqOrig_i$. In the same way we define the bottom,
top, and right neighborhoods, denoted by $N_B(\SqOrig_i),
N_T(\SqOrig_i)$, and $N_R(\SqOrig_i)$,
respectively.\index{neighborhood}\index{neighborhood!left}
\index{neighborhood!right}\index{neighborhood!bottom}\index{neighborhood!top}

A point, $P$, is called {\em unsupported}, if there is a vertical
line segment pointing from $P$ to the bottom of $S$ whose interior
lies completely inside a hole. Otherwise, $P$ is {\em supported}. A
section of a line segment is supported, if every point in this
section is supported.

For an object $\xi$, we refer to the boundary as $\partial \xi$, to
the interior as $\xi^\circ$, and to its area by $|\xi|$. If $\xi$ is
a line segment, then $|\xi|$ denotes its length.

\begin{figure}[t]
\psfrag{Ait}{$\SqOrig_i$} %
\psfrag{sky}{${\cal S}_{\SqOrig_i}$} %
\psfrag{leftS}{$L_S$} %
\psfrag{bS}{$B_S$} %
\psfrag{leftseq}{${\cal L}_{\SqOrig_i}$} %
\psfrag{bseq}{${\cal B}_{\SqOrig_i}$} %
\psfrag{left}{$L_{\SqOrig_i}$} %
\psfrag{rig}{$R_{\SqOrig_i}$} %
\psfrag{top}{$T_{\SqOrig_i}$} %
\psfrag{X}{$B_{\SqOrig_i}$} %
\centering
\includegraphics[height=6.5cm]{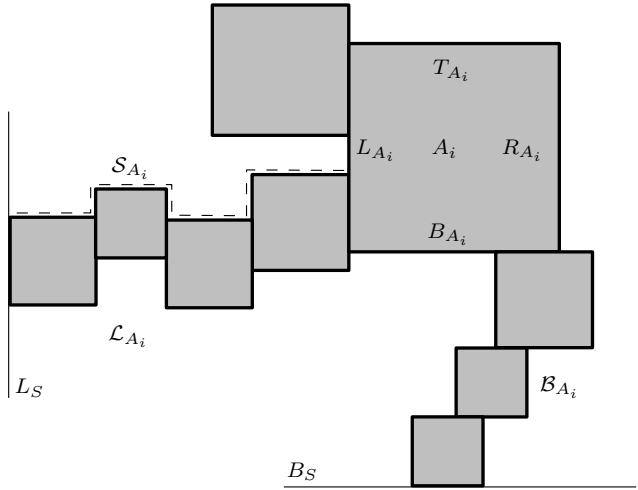}
\caption{The square $\SqOrig_i$ with its left sequence ${\cal
L}_{\SqOrig_i}$, the bottom sequence ${\cal B}_{\SqOrig_i}$, and the
skyline ${\cal S}_{\SqOrig_i}$. The left sequence ends at the left
side of $S$, and the bottom sequence at the bottom side of $S$.}
\label{fig:defs_psfrag}
\end{figure}

\paragraph{Outline of the Analysis}
We proceed as follows: First, we state some basic properties of the
generated packing (Section~\ref{sect:basicproperties}). In
Section~\ref{sect:modify} we simplify the shape of the holes by
partitioning a hole, produced by BottomLeft, into several disjoint
new holes. In the packing, these new holes are open at their top
side, so we introduce {\em virtual lids} that close these holes.
Afterwards, we estimate the area of a hole in terms of the squares
that enclose the hole (Section~\ref{sect:holesize}). First, we bound
the area of holes that have no virtual lid and whose boundary does
not intersect the boundary of the strip.
Then, we analyze
holes with a virtual lid
; as it turns out, these are ``cheaper'' than holes with non-virtual
lids. Finally, we show that holes that touch the strip's boundary
are just a special
case. 
Section~\ref{sect:summingupcharges} summarizes the costs that are
charged to a square.

\subsection{Basic Properties of the Generated Packing\label{sect:basicproperties}}
In this section, we show some basic properties of a packing
generated by BottomLeft. In particular, we analyze structural
properties of the boundary of a hole.

We say that a square, $\SqOrig_i$, {\em contributes} to the boundary
of a hole, $H_h$, iff $\partial \SqOrig_i$ and $\partial H_h$
intersect in more than one point, {\em i.e.,} $|\partial \SqOrig_{i}
\cap \partial H_h| > 0$. For convenience, we denote the squares on
the boundary of a hole by $\SqLoch_1, \ldots, \SqLoch_k$ in
counterclockwise order starting with the
upper left square; see Fig.~\ref{fig:example_bn}. It is always clear
from the context which hole defines this sequence of squares. Thus,
we chose not to introduce an additional superscript referring to the
hole. We define $\SqLoch_{k+1} = \SqLoch_1$, $\SqLoch_{k+2} =
\SqLoch_2$, and so on. By $P_{i,i+1}$ we denote the point where
$\partial H_h$ leaves the boundary of $\SqLoch_i$ and enters the
boundary of $\SqLoch_{i+1}$; see Fig.~\ref{fig:bottomseq}.

Let $\SqOrig_i$ be a square packed by BottomLeft. Then $\SqOrig_i$
can be moved neither to the left nor down. This implies that either
$N_L(\SqOrig_i)\neq \emptyset$ ($N_B(\SqOrig_i)\neq \emptyset$) or
that $L_{\SqOrig_i}$ ($B_{\SqOrig_i}$) coincides with $L_S$ ($B_S$).
Therefore, the following two sequences ${\cal L}_{\SqOrig_i}$ and
${\cal B}_{\SqOrig_i}$ exist: The first element of ${\cal
L}_{\SqOrig_i}$ (${\cal B}_{\SqOrig_i}$) is ${\SqOrig_i}$. The next
element is chosen as an arbitrary left (bottom) neighbor of the
previous element. The sequence ends if no such neighbor exits. We
call ${\cal L}_{\SqOrig_i}$ the {\em left sequence} and ${\cal
B}_{\SqOrig_i}$ the {\em bottom sequence} of a square ${\SqOrig_i}$;
see Fig.~\ref{fig:defs_psfrag} \index{bottom sequence}\index{left
sequence}

We call the polygonal chain from the upper right corner of the first
element of ${\cal L}_{\SqOrig_i}$ to the upper left corner of the
last element, while traversing the boundary of the sequence in
counterclockwise order, the {\em skyline}, ${\cal S}_{\SqOrig_i}$,
of ${\SqOrig_i}$.


\begin{figure}
\psfrag{A1t}{$\tilde{A}_1$} %
\psfrag{A2t}{$\tilde{A}_2$} %
\psfrag{A3t}{$\tilde{A}_3$} %
\psfrag{Akt}{$\tilde{A}_k$} %
\psfrag{Akm1t}{$\tilde{A}_{k-1}$} %
\psfrag{Hh}{$H_h$} %
\psfrag{Hh1}{$H_{h+1}$} %
\psfrag{Hh2}{$H_{h+2}$} %
\psfrag{Hhstar}{$H_h^\star$} %
\psfrag{Hstarhp1}{$H^\star_{h+1}$} %
\psfrag{Dhl}{$D^h_l$} %
\psfrag{Dhr}{$D^h_r$} %
\psfrag{Dhp1l}{$D^{h+1}_l$} %
\psfrag{starDhr}{$^\star \! D^h_r$} %
\psfrag{starDhl}{$^\star \! D^h_l$} %
\centering
\includegraphics[height=6.5cm]{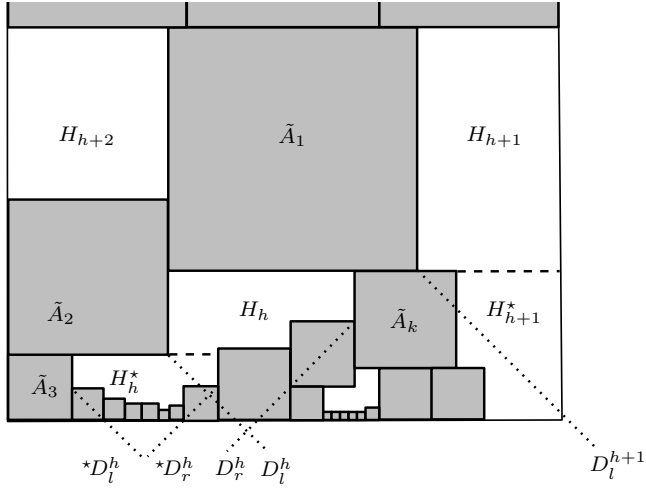}
\caption{A packing produced by BottomLeft. The squares
$\SqLoch_1,\ldots,\SqLoch_k$ contribute to the boundary of the hole
$H_h$. In the analysis, $H_h$ is split into a number of subholes. In
the shown example one new subhole $H_h^\star$ is created. Note that
the square $\SqLoch_1$ also contributes to the holes $H_{h+1}$ and
$H_{h+2}$. Moreover, it serves as a virtual lid for $H_{h+1}^\star$.
} \label{fig:example_bn}
\end{figure}

Obviously, ${\cal S}_{\SqOrig_i}$ has an endpoint on $L_S$ and
${\cal S}_{\SqOrig_i}^\circ \cap H_h^\circ = \emptyset$.
With the help of ${\cal L}_{\SqOrig_i}$ and ${\cal B}_{\SqOrig_i}$
we can prove (see Fig.~\ref{fig:bottomseq}):

\begin{lemma}\label{lem:directionoftraversal}
Let ${\SqLoch_i}$ be a square that contributes to $\partial H_h$.
Then, \begin{list}{}{}
\item[(i)] $\partial H_h \cap \partial {\SqLoch_i}$ is a single curve, and
\item[(ii)] if $\partial H_h$ is traversed in counterclockwise (clockwise) order,
$\partial H_h \cap \partial {\SqLoch_i}$ is traversed in clockwise
(counterclockwise) order w.r.t.\ $\partial {\SqLoch_i}$.
\end{list}
\end{lemma}

\begin{proof}
For the first part, suppose for a contradiction that $\partial H_h
\cap \partial {\SqLoch_i}$ consists of at least two curves, $c_1$
and $c_2$. Consider a simple curve, $C$, that lies completely inside
$H_h$ and has one endpoint in $c_1$ and the other one in $c_2$. We
add the straight line between the endpoints to $C$ and obtain a
simple closed curve $C'$. As $c_1$ and $c_2$ are not connected,
there is a square, $\SqLoch_j$, inside $C'$ that is a neighbor of
${\SqLoch_i}$. If ${\SqLoch_j}$ is a left, right or bottom neighbor
of ${\SqLoch_i}$ this contradicts the existence of ${\cal
B}_{\SqLoch_j}$ and if it is a top neighbor this contradicts the
existence of ${\cal L}_{\SqLoch_j}$. Hence, $\partial H_h \cap
\partial {\SqLoch_i}$ is a single curve.

For the second part, imagine that we walk along $\partial H_h$ in
counterclockwise order. Then, the interior of $H_h$ lies on our
left-hand side, and all squares that contribute to $\partial H_h$
lie on our right-hand side. Hence, their boundaries are traversed in
clockwise order w.r.t.\ their interior.
\end{proof}

We define $P$ and $Q$ to be the left and right endpoint,
respectively, of the line segment $\partial \SqLoch_1 \cap \partial
H_h$.
Two squares $\SqLoch_i$ and $\SqLoch_{i+1}$ can basically be
arranged in four ways, {\em i.e.,} $\SqLoch_{i+1}$ can be a left,
right, bottom or top neighbor of $\SqLoch_i$. The next lemma
restricts these possibilities:


\begin{lemma}\label{lem:pairsofsquares_1}
Let $\SqLoch_{i}$, $\SqLoch_{i+1}$ be a pair of squares that
contribute to the boundary of a hole $H_h$.
\begin{list}{}{\setlength{\itemsep}{3pt}}
\item[(i)]
If $\SqLoch_{i+1} \in N_L(\SqLoch_{i})$, then either $\SqLoch_{i+1}
= \SqLoch_1$ or $\SqLoch_{i} = \SqLoch_1$.
\item[(ii)]
If $\SqLoch_{i+1} \in N_T(\SqLoch_{i})$, then $\SqLoch_{i+1} =
\SqLoch_1$ or $\SqLoch_{i+2} = \SqLoch_1$.
\end{list}
\end{lemma}

\begin{proof}
(i) Let $\SqLoch_{i+1} \in N_L(\SqLoch_{i})$. Consider the endpoints
of the vertical line $R_{\SqLoch_{i+1}} \cap L_{\SqLoch_{i}}$; see
Fig.~\ref{fig:bottomseq}. We traverse $\partial H_h$ in
counterclockwise order starting in $P$. By Lemma
\ref{lem:directionoftraversal}, we traverse $\partial \SqLoch_i$ in
clockwise order, and therefore, $P_{i,i+1}$ is the lower endpoint of
$R_{\SqLoch_{i+1}} \cap L_{\SqLoch_{i}}$. Now, ${\cal
B}_{\SqLoch_{i}}$, ${\cal B}_{\SqLoch_{i+1}}$, and the segment of
$B_S$
completely enclose an area that completely contains the hole, $H_h$.
If the sequences have a square in common, we consider the area
enclosed up to the first intersection.
Therefore, if $b_{\SqLoch_i+1}\geq b_{\SqLoch_{i}}$ then
$\SqLoch_{i+1}=\SqLoch_1$ else $\SqLoch_{i}=\SqLoch_1$ by the
definition of $\overline{PQ}$.

The proof of (ii) follows almost directly from the first part. Let
$\SqLoch_{i+1} \in N_T(\SqLoch_{i})$. If $\partial H_h$ is traversed
in counterclockwise order, we know that $\partial \SqLoch_{i+1}$ is
traversed in clockwise order, and we know that $\SqLoch_{i+1}$ must
be supported to the left. Therefore, $\SqLoch_{i+2} \in
N_L(\SqLoch_{i+1}) \cup N_B(\SqLoch_{i+1})$. Using the first part of
the lemma, we conclude that, if $\SqLoch_{i+2} \in
N_L(\SqLoch_{i+1})$ then $\SqLoch_{i+2} = \SqLoch_{1}$ or
$\SqLoch_{i+1} = \SqLoch_{1}$, or if $\SqLoch_{i+2} \in
N_B(\SqLoch_{i+1})$ then $\SqLoch_{i+1} = \SqLoch_{1}$.
\end{proof}

\begin{figure}
\psfrag{Ait}{$\SqLoch_i$} %
\psfrag{Aip1t}{$\SqLoch_{i+1}$} %
\psfrag{Aip2t}{$\SqLoch_{i+2}$} %
\psfrag{P}{$P$} %
\psfrag{Hh}{$H_h$} %
\psfrag{Piip1}{$P_{i,i+1}$} %
\psfrag{RL}{$R_{\SqLoch_{i+1}} \! \cap L_{\SqLoch_{i}}$} %
\psfrag{BotAip1}{${\cal B}_{\tilde{A}_{i+1}}$} %
\psfrag{BotAi}{${\cal B}_{\tilde{A}_{i}}$} %
\centering
\includegraphics[height=6cm]{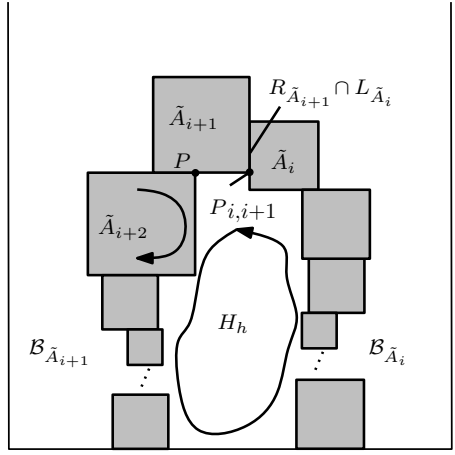}
\caption{The hole $H_h$ with the two squares $\SqLoch_i$ and
$\SqLoch_{i+1}$ and their bottom sequences. In this situation,
$\SqLoch_{i+1}$ is $\SqLoch_1$. If $\partial H_h$ is traversed in
counterclockwise order then $\partial H_h \cap
\partial \tilde{A}_{i+2}$ is traversed in clockwise order w.r.t.\ to
$\partial \tilde{A}_{i+2}$.} \label{fig:bottomseq}
\end{figure}

The last lemma implies that either $\SqLoch_{i+1} \in
N_B(\SqLoch_{i})$ or $\SqLoch_{i+1} \in N_R(\SqLoch_{i})$ holds for
all $i=2,\ldots,k-2$; see Fig.~\ref{fig:example_bn}. The next lemma
shows that there are only two possible arrangements of the squares
$\SqLoch_{k-1}$ and $\SqLoch_k$:

\begin{lemma}
Either $\SqLoch_{k} \in N_R(\SqLoch_{k-1})$ or $\SqLoch_{k} \in
N_T(\SqLoch_{k-1})$.
\end{lemma}

\begin{proof}
We traverse $\partial H_h$ from $P$ in clockwise order. From the
definition of $\overline{PQ}$ and
Lemma~\ref{lem:directionoftraversal} we know that $P_{k,1}$ is a
point on $L_{\SqLoch_k}$. If $P_{k-1,k} \in L_{\SqLoch_k}$, then
$\SqLoch_{k} \in N_R(\SqLoch_{k-1})$; if $P_{k-1,k} \in
B_{\SqLoch_k}$, then $\SqLoch_{k} \in N_T(\SqLoch_{k-1})$. In any
other case $\SqLoch_k$ does not have a bottom neighbor.
\end{proof}

Following the distinction described in the lemma, we say that a {\bf
hole} is of {\bf Type~I} if $\SqLoch_{k} \in N_R(\SqLoch_{k-1})$,
and of {\bf Type~II} if $\SqLoch_{k} \in N_T(\SqLoch_{k-1})$; see
Fig.~\ref{fig:typeIandtypeII}.

\subsection{Splitting Holes\label{sect:modify}}\index{hole!splitting}
Let $H_h$ be a hole whose boundary does not touch the boundary of
the strip, {\em i.e.,} the hole is completely enclosed by squares.
We define two lines that are essential for the computation of an
upper bound for the area of a hole, $H_h$: The {\em left diagonal},
\index{left diagonal} $D_l^h$, is defined as the straight line with
slope $-1$ starting in $P_{2,3}$ if $P_{2,3} \in R_{\SqLoch_2}$ or,
otherwise, in the lower right corner of $\SqLoch_2$; see
Fig.~\ref{fig:typeIandtypeII}. We denote the point where $D_l^h$
starts by $P'$. The {\em right diagonal},\index{right diagonal}
$D_r^h$, is defined as the line with slope $1$ starting in
$P_{k-1,k}$ if $\SqLoch_{k} \in N_R(\SqLoch_{k-1})$ (Type~I) or in
$P_{k-2,k-1}$, otherwise (Type~II). Note that $P_{k-2,k-1}$ lies on
$L_{\SqLoch_{k-1}}$ because otherwise, there would not be a left
neighbor of $\SqLoch_{k-1}$. We denote the point where $D_r^h$
starts by $Q'$. If $h$ is clear or does not matter we omit the
superscript.


\begin{lemma}\label{lem:dr}
Let $H_h$ be a hole and $D_r$ its right diagonal. Then, $D_r \cap
H_h^{\circ} = \emptyset$.
\end{lemma}
\begin{proof}
Consider the left sequence, ${\cal L}_{\SqLoch_k} =
(\SqLoch_k=\alpha_1, \alpha_2, \ldots)$ or ${\cal
L}_{\SqLoch_{k-1}}= (\SqLoch_{k-1}=\alpha_1, \alpha_2, \ldots)$, for
$H_h$ being of Type~I or II, respectively. It is easy to show by
induction that the upper left corners of the $\alpha_i$'s lie above
$D_r$: If $D_r$ intersects $\partial\alpha_i$ at all, the first
intersection is on $R_{\alpha_i}$, the second on $B_{\alpha_i}$.
Thus, at least the skyline separates $D_r$ and $H_h$.
\end{proof}

\begin{figure}[t!]
\psfrag{Aim1t}{$\SqLoch_{i-1}$} %
\psfrag{Ait}{$\SqLoch_{i}$} %
\psfrag{Apt}{$\SqLoch_{p}$} %
\psfrag{Aip1t}{$\SqLoch_{i+1}$} %
\psfrag{Aqt}{$\SqLoch_{q}$} %
\psfrag{LAqt}{$L_{\SqLoch_q}$} %
\psfrag{Dl}{$D_l$} %
\psfrag{H1h}{$H_h^{(1)}$} %
\psfrag{Hstarh}{$H_h^\star$} %
\psfrag{pHh}{$\partial H_h$} %
\psfrag{M}{$M$} %
\psfrag{N}{$N$} %
\psfrag{F}{$F$} %
\psfrag{E}{$E$} %
\psfrag{VN}{$V_N$} %
\psfrag{Pim1i}{$P_{i-1,i}$} %
\psfrag{Piip1}{$P_{i,i+1}$} %
\psfrag{casea}{Case A} %
\psfrag{caseb}{Case B} %
\centering
\includegraphics[height=6.0cm]{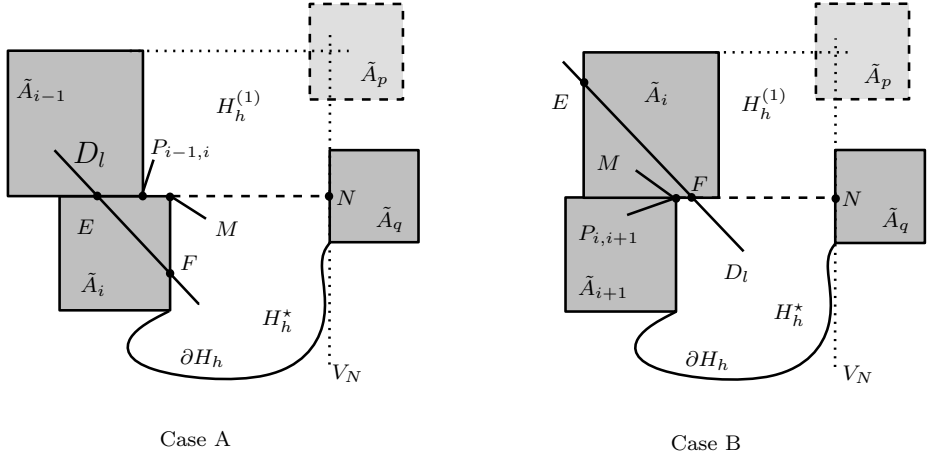}
\caption{$D_l$ can intersect $\SqLoch_i$ (for the second time) in
two different ways: on the right side or on the bottom side. In
Case~A, the square $\SqLoch_{i-1}$ is on top of $\SqLoch_i$; in
Case~B, $\SqLoch_{i}$ is on top of $\SqLoch_{i+1}$.}
\label{fig:dlintersection}
\end{figure}

If Lemma \ref{lem:dr} would also hold for $D_l$, we could use the
polygon formed by $D_l$, $D_r$, and the part of the boundary of
$H_h$ between $Q'$ and $P'$ to bound the area of $H_h$,
but---unfortunately---it does not.

Let $F$ be the first nontrivial intersection point of $\partial H_h$
and $D_l$, while traversing $\partial H_h$ in counterclockwise
order, starting in $P$. $F$ is on the boundary of a square,
$\SqLoch_i$. Let $E$ be the other intersection point of $D_l$ and
$\partial \SqLoch_i$.

It is a simple observation that if $D_l$ intersects a square,
$\SqLoch_i$, in a nontrivial way, {\em i.e.,} in two different
points, $E$ and $F$,
then either $F \in R_{\SqLoch_i}$ and $E \in T_{\SqLoch_i}$ or $F
\in B_{\SqLoch_i}$ and $E \in L_{\SqLoch_i}$. To break ties, we
define that an intersection in the lower right corner of $\SqLoch_i$
belongs to $B_{\SqLoch_i}$.
Now, we split our hole, $H_h$, into two new holes, $H_h^{(1)}$ and
$H_h^\star$. We consider two cases (see
Fig.~\ref{fig:dlintersection}):
\begin{list}{}{}

\item[$\bullet$] Case A: $F \in R_{\SqLoch_i} \setminus
B_{\SqLoch_i}$

\item[$\bullet$] Case B: $F \in B_{\SqLoch_i}$

\end{list}

In Case~A, we define $\SqLoch_{\mathrm{up}} := \SqLoch_{i-1}$ and
$\SqLoch_{\mathrm{low}} := \SqLoch_i$, in Case~B
$\SqLoch_{\mathrm{up}} := \SqLoch_{i}$ and $\SqLoch_{\mathrm{low}}
:= \SqLoch_{i+1}$. Observe the horizontal ray that emanates from the
upper right corner of $\SqLoch_{\mathrm{low}}$ to the right: This
ray is subdivided into supported and unsupported sections. Let
$U=\overline{MN}$ be the leftmost unsupported section with left
endpoint $M$ and right endpoint $N$; see
Fig.~\ref{fig:dlintersection}. Now, we split $H_h$ into two parts,
$H_h^\star$ below $\overline{MN}$ and $H_h^{(1)} := H_h \backslash
H_h^\star$.




We split $H_h^{(1)}$ into $H_h^{(2)}$ and $H_h^{\star\star}$ etc.,
until there is no further intersection between the boundary of
$H_h^{(z)}$ and $D_l^h$. Because there is a finite number of
intersections, this process will eventually terminate. In the
following, we show that $H_h^{(1)}$ and $H_h^\star$ are indeed two
separate holes, and that $H_h^\star$ has the same properties as an
original one, {\em i.e.,} it is a hole of Type~I or II. Thus, we can
analyze $H_h^\star$ using the same technique, {\em i.e.,} we may
split $H_h^\star$ w.r.t.\ {\em its} left diagonal. We need some
lemmas for this proof:


\begin{lemma} \label{lem:caseAcaseB_1}
Using the above notation we have $\SqLoch_{\mathrm{low}} \in
N_B(\SqLoch_{\mathrm{up}})$.
\end{lemma}

\begin{proof}
We consider Case~A and Case~B separately, starting with Case A.
We traverse $\partial H_h$ from $F$ in clockwise order. By Lemma
\ref{lem:directionoftraversal}, $\SqLoch_{i-1}$ is the next square
that we reach; see Fig.~\ref{fig:dlintersection}. Because $F$ is the
first
intersection, $P_{i-1,i}$ lies 
between $F$ and $E$. 
Thus, either $P_{i-1,i} \in T_{\SqLoch_i}$ or $P_{i-1,i} \in
R_{\SqLoch_i}$ holds. With Lemma~\ref{lem:pairsofsquares_1}, the
latter implies either $\SqLoch_{i-1}=\SqLoch_1$ or
$\SqLoch_{i}=\SqLoch_1$. Because $b_{\SqLoch_{i-1}} >
b_{\SqLoch_{i}}$ holds, only $\SqLoch_{i-1} = \SqLoch_1$ is
possible, and therefore, $\SqLoch_i = \SqLoch_2$. $D_l$ intersects
$\SqLoch_2$ in the lower left corner---which is not included in this
case---or in $P_{2,3}$. However, $P_{2,3}\in R_{\SqLoch_2}$ cannot
be an intersection, because this would imply $\SqLoch_3 \in
N_R(\SqLoch_2)$. Thus, only $P_{i-1,i} \in T_{\SqLoch_i}$ is
possible.

In Case~B,
we traverse $\partial H_h$ from $F$ in counterclockwise order, and
$\SqLoch_{i+1}$ is the next square that we reach. Because $F$ is the
first intersection, it follows that
$P_{i,i+1}$ 
lies on $\partial \SqLoch_i$ between $F$ and $E$ in clockwise order;
see Fig.~\ref{fig:dlintersection}. Thus, $\SqLoch_{i+1} \in
N_B(\SqLoch_i)$ or $\SqLoch_{i+1} \in N_L(\SqLoch_i)$ holds. If
$\SqLoch_{i+1} \in N_L(\SqLoch_i)$, we have $P_{i,i+1} \in
L_{\SqLoch_i}$. If we move from $P_{i,i+1}$ to $F$ on $\partial
\SqLoch_i$, we move in clockwise order on $\partial H_h$. If we
reach $P_{i-1,i}$ before $F$, the square, $\SqLoch_{i-1}$, is
between $P_{i,i+1}$ and $F$. The points, $P_{i,i+1}$ and $F$, are on
$\partial H_h$, and thus, $\partial H_h \cap \SqLoch_i$ is
disconnected, which contradicts
Lemma~\ref{lem:directionoftraversal}. Thus, we reach $F$ before
$P_{i-1,i}$. Moreover, $\SqLoch_i$ must have a bottom neighbor, and
therefore, $P_{i-1,i} \in
B_{\SqLoch_i}^\circ$. 
By Lemma \ref{lem:pairsofsquares_1}, we have $\SqLoch_i = \SqLoch_1$
or $\SqLoch_{i+1}=\SqLoch_1$. Both cases contradict the fact that
$D_l$ intersects neither $\SqLoch_2$ in the lower right corner nor
$\SqLoch_1$. Altogether, $P_{i,i+1}$ must be on $B_{\SqLoch_i}$ to
the left of $F$.
\end{proof}

The last lemma states that in both cases, there are two squares for
which one is indeed placed on top of the other.
\begin{lemma} \label{lem:caseAcaseB_2}
$M$ is the upper right corner of $\SqLoch_{\mathrm{low}}$.
\end{lemma}
\begin{proof}
Case~A: We know $F\in R_{\SqLoch_i}$ and $P_{i-1,i} \in
T_{\SqLoch_i}$. By Lemma~\ref{lem:directionoftraversal}, the upper
right corner, $M'$, of $\SqLoch_i$ belongs to $\partial H_h$.
Because $F$ does not coincide with $M'$ (degenerate intersection),
$\overline{FM'}$ is a vertical line of positive length. Hence, $M'$
is the beginning of an unsupported section of the horizontal ray
emanating from $M'$ to the right. Thus, the first unsupported
section starts in $M'$; that is, $M=M'$. A similar argument holds in
Case~B.
\end{proof}

To ensure that $H_h^\star$ is well defined, we show that it has a
closed boundary. Obviously, $\overline{MN}$ and the part of
$\partial H_h$ counterclockwise from $M$ to $N$ forms a closed
curve. We place an imaginary copy of $\SqLoch_{\mathrm{up}}$ on
$\overline{MN}$, such that the lower right corner is placed in $N$.
We call the copy the {\em virtual lid}\index{virtual lid}, denoted
by $\SqLoch_{\mathrm{up}}'$. We show that $\overline{MN} <
\SqLochklein_{\mathrm{up}}$ holds, where
$\SqLochklein_{\mathrm{up}}$ denotes the side length of
$\SqLoch_{\mathrm{up}}$. Thus, $\overline{MN}$ is completely covered
by the virtual copy of $\SqLoch_{\mathrm{up}}$, and in turn, we can
choose the virtual block as a new lid for $H_h^\star$.

\begin{lemma}\label{lem:sizeofMN}
With the above notation we have $\overline{MN} <
\SqLochklein_{\mathrm{up}}$.
\end{lemma}

\begin{proof}
We show that at the time $\SqLoch_{\mathrm{up}}$ is packed by
BottomLeft, it can be moved to the right along $\overline{MN}$, such
that the lower right corner coincides with $N$. Since
$\overline{MN}$ is unsupported, $\overline{MN} \geq
\SqLochklein_{\mathrm{up}}$ implies that there would have been a
position for $\SqLoch_{\mathrm{up}}$ that is closer to the bottom of
$S$ than its current position.

Let $V_N$ be the vertical line passing through the point $N$, and
let $v_N$ be its $x$-coordinate. Assume that there is a square,
$\SqLoch_p$, that prevents $\SqLoch_{\mathrm{up}}$ from being moved.
Then, $\SqLoch_p$ fulfills $l_{\SqLoch_p} < v_N$ and $b_{\SqLoch_p}
< t_{\SqLoch_{\mathrm{up}}}$ $(*)$; see
Fig.~\ref{fig:dlintersection}. Now, consider the sequence ${\cal
L}_{\SqLoch_p}$, and note that all squares in ${\cal L}_{\SqLoch_p}$
are placed before $\SqLoch_{\mathrm{up}}$. From $(*)$ we conclude
that the skyline, ${\cal S}_{\SqLoch_p}$, may intersect the
horizontal line passing through $T_{\SqLoch_{\mathrm{low}}}$ only to
the left of $v_N$. If the skyline intersects or touches in
$\overline{MN}$, we have a contradiction to the choice of $M$ and
$N$ as endpoints of the first unsupported section. An intersection
between $M$ and $P_{\mathrm{up},{\mathrm{low}}}$ is not possible,
because this part completely belongs to
$T_{\SqLoch_{\mathrm{low}}}$. Therefore, ${\cal S}_{\SqLoch_p}$
either intersects the horizontal line to the left of
$P_{\mathrm{up},{\mathrm{low}}}$ or it reaches $L_S$ before. This
implies that $\SqLoch_{\mathrm{up}}$ must pass $\SqLoch_p$ on the
right side and at the bottom side to get to its final position. In
particular, $b_{\SqLoch_p} < t_{\SqLoch_{\mathrm{up}}}$ implies that
$\SqLoch_{\mathrm{up}}$'s path must go upwards to reach its final
position; such a path contradicts the choice of BottomLeft.
\end{proof}

Using the preceding lemmas, we can prove the following:

\begin{corollary}\label{cor:hstar}
Let $H_h^\star$ and $\SqLoch'_{\mathrm{up}}$ be defined as above.
$H_h^\star$ is a hole of Type~I or Type~II with virtual lid
$\SqLoch'_{\mathrm{up}}$.
\end{corollary}

\begin{proof}
$H_h^\star$ has a closed boundary, and there is at least a small
area below $\overline{MN}$ in which no squares are placed. Hence,
$H_h^\star$ is a hole. Using the arguments that the interior of
$\overline{MN}$ is unsupported and that $N$ is supported and lies on
$L_{\SqLoch_q}$, for some $1 \leq q \leq k$, we conclude that there
is a vertical line of positive length below $N$ on $\partial
\SqLoch_q$ that belongs to $\partial H_h$. If we move from $N$ on
$\partial \SqLoch_q$ in counterclockwise order, we move on $\partial
H_h$ in clockwise order and reach $\SqLoch_{q-1}$ next. If
$P_{q-1,q} \in L_{\SqLoch_q}$, then $H_h^\star$ is of Type~I. If
$P_{q-1,q} \in B_{\SqLoch_q}$, then it is of Type~II. $P_{q-1,q}
\notin L_{\SqLoch_q} \cup B_{\SqLoch_q}$ yields a contradiction,
because in this case there is no bottom neighbor for $\SqLoch_q$.
$\SqLoch'_{\mathrm{up}}$ is the unique lid by the existence of the
sequences ${\cal B}_{\SqLoch_q}$ and ${\cal
B}_{\SqLoch_{\mathrm{low}}}$.
\end{proof}

Note that the preceding lemmas also hold for the holes
$H_h^{(...)}$, $H_h^{\star\star}$, $H_h^{\star\star\star}$, and so
on.

\begin{lemma}
For every square, $\SqOrig_i$, there is at most one copy of
$\SqOrig_i$.
\end{lemma}
\begin{proof}
A square, $\SqOrig_i$, is used as a virtual lid, only if its lower
right corner is on the boundary of the hole that is split. Because
its corner can be on the boundary of at most one hole, there is only
one hole with virtual lid $\SqOrig_i$.
\end{proof}


\subsection{Computing the Area of a Hole\label{sect:holesize}}
\index{hole!computing the area} In this section we show how to
compute the area of a hole. In the preceding section we eliminated
all intersections of $D_l^h$ with the boundary of the hole,
$H_h^{(z)}$, by splitting the hole. Thus, we assume that we have a
set of holes, $\hat{H}_h$, $h=1,\ldots, s^\star$, that fulfill
$\partial \hat{H}_h \cap D_l^h=\emptyset$ and have either a
non-virtual or a virtual lid.

Our aim is to bound $|\hat{H}_h|$ by the areas of the squares that
contribute to $\partial \hat{H}_h$. A square, $\SqOrig_i$, may
contribute to more than one hole. Therefore, it is to expensive to
use its total area, $a_i^2$, in the bound for a single hole.
Instead, we charge only fractions of $a_i^2$ per hole. Moreover, we
charge every edge of $\SqOrig_i$ separately. By
Lemma~\ref{lem:directionoftraversal}, $\partial \hat{H}_h \cap
\partial \SqOrig_i$ is connected. In particular, every side of
$\SqOrig_i$ contributes at most one (connected) line segment to
$\partial \hat{H}_h$. For the left (bottom, right) side of a square,
$\SqOrig_i$, we denote the length of the line segment contributed to
$\partial\hat{H}_h$ by $\lambda_i^h$ ($\beta_i^h$, $\rho_i^h$;
respectively). If a side of a square does not contribute to a hole,
the corresponding length of the line segment is defined to be zero.

Let $c_{h,i}^{\{\lambda,\beta,\rho\}}$ be appropriate coefficients,
such that the area of a hole can be charged against the area of the
adjacent squares, {\em i.e.,}
\begin{equation*}
|\hat{H}_h| \leq \sum_{i=1}^{n+1} c_{h,i}^\lambda (\lambda_i^h)^2 +
c_{h,i}^\beta (\beta_i^h)^2 + c_{h,i}^\rho (\rho_i^h)^2\,.
\end{equation*}

As each point on $\partial \SqOrig_i$ is---obviously---on the
boundary of at most one hole, the line segments are pairwise
disjoint. Thus, for the left side of $\SqOrig_i$, the two squares
inside $\SqOrig_i$ induced by the line segments, $\lambda_i^h$ and
$\lambda_i^g$, of two different holes, $\hat{H}_h$ and $\hat{H}_g$,
do not overlap. 
Therefore, we obtain
\begin{equation*}
\sum_{h=1}^{s^\star}  
c_{h,i}^\lambda \cdot (\lambda_i^h)^2 \leq c_i^\lambda \cdot a_i^2\;
,
\end{equation*}
where $c_i^\lambda$ is the maximum of the $c_{h,i}^\lambda$'s taken
over all holes $\hat{H}_h$. We call $c_i^\lambda$ the {\em charge of
$L_{\SqOrig_i}$} and define $c_i^\beta$ and $c_i^\rho$ analogously.

We use virtual copies of some squares as lids. However, for every
square, $\SqOrig_i$, there is at most one copy, $\SqOrig_i'$. We
denote the line segments and charges corresponding to $\SqOrig_i'$
by $\lambda_{i'}^h$, $c_{h,i'}^\lambda$, and so on. Taking the
charges to the copy into account, the {\em total charge of
$\SqOrig_i$} is given by

\begin{equation*}
c_i = c_i^\lambda + c_i^\beta + c_i^\rho + c_{i'}^\lambda +
c_{i'}^\beta + c_{i'}^\rho\,.
\end{equation*}
Altogether, we bound the total area of the holes by
\begin{equation*}
\sum_{h=1}^{s^\star} |\hat{H}_h| \leq \sum_{i=1}^{n+1} c_i\cdot
a_i^2 \leq \sum_{i=1}^{n+1} c\cdot a_i^2\, ,
\end{equation*}
where $c=\max_{i=1,\ldots,n} \{c_i\}$. In the following, we want to
find an upper bound on $c$.

\paragraph{Holes with a Non-Virtual Lid\label{sect:nonvirtual}}
We know that each hole is either of Type~I or~II. Moreover, we
removed all intersections of $\hat{H}_h$ with its diagonal,~$D_l^h$.
Therefore, $\hat{H}_h$ lies completely inside the polygon formed by
$D_l^h$, $D_r^h$, and the part of $\partial \hat{H}_h$ that is
clockwise between $P'$ and $Q'$; see Fig.~\ref{fig:typeIandtypeII}.

\begin{figure}
\psfrag{Q}{$Q$} %
\psfrag{Qprime}{$Q'$} %
\psfrag{P}{$P$} %
\psfrag{Pprime}{$P'$} %
\psfrag{ti}{Type~I} %
\psfrag{tii}{Type~II} %
\psfrag{A1t}{$\SqLoch_1$} %
\psfrag{A2t}{$\SqLoch_2$} %
\psfrag{A3t}{$\SqLoch_3$} %
\psfrag{Akt}{$\SqLoch_k$} %
\psfrag{Akm1t}{$\SqLoch_{k-1}$} %
\psfrag{Akm2t}{$\SqLoch_{k-2}$} %
\psfrag{Dl}{$D_l$} %
\psfrag{Dr}{$D_r$} %
\psfrag{R1}{$R_1$} %
\psfrag{R2}{$R_2$} %
\psfrag{Delta1}{$\Delta_1$} %
\psfrag{Delta2}{$\Delta_2$} %
\psfrag{a}{$V_Q$} %
\psfrag{rho2h}{$\rho_2^h$} %
\psfrag{x}{$\beta_1^h$} %
\psfrag{y}{$\beta_k^h$} %
\psfrag{z}{$\lambda_{k-1}^h$} %
\psfrag{Pkm1k}{$P_{k-1,k}$} %
%
\includegraphics[height=5.9cm]{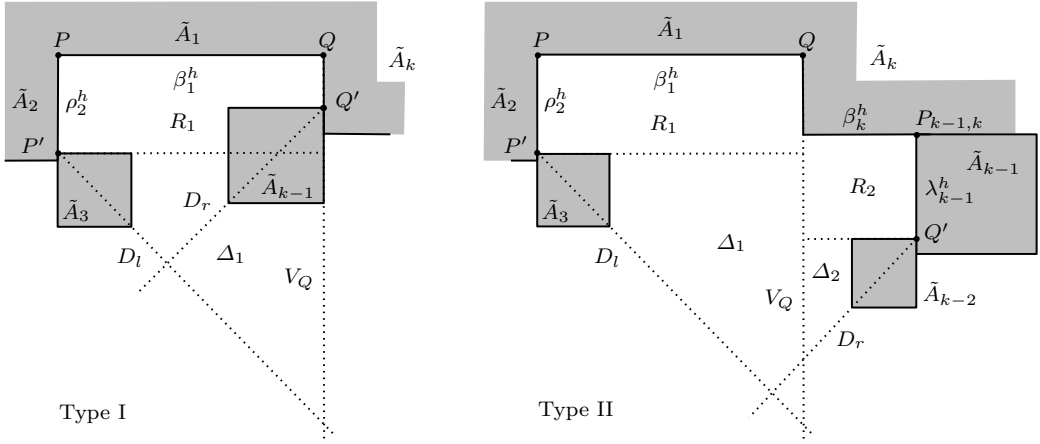}
\caption{Holes of Type~I and~II with their left and right
diagonals.} \label{fig:typeIandtypeII}
\end{figure}

If $\hat{H}_h$ is of Type~I, then we consider the rectangle, $R_1$,
of area $\rho_2^h \cdot \beta_1^h$ induced by the points $P$, $P'$,
and $Q$. Moreover, let $\Delta_1$ be the triangle below $R_1$ formed
by the bottom side of $R_1$, $D_l^h$, and the vertical line, $V_Q$,
passing through $Q$; see Fig.~\ref{fig:typeIandtypeII}. We obtain:

\begin{lemma}\label{cor:sizecatI}
Let $\hat{H}_h$ be a hole of Type~I. Then,
$$
|\hat{H}_h| \leq (\beta_1^h)^2 + \frac12 (\rho_2^h)^2\; .
$$

\end{lemma}
\begin{proof}
Obviously, $|\hat{H}_h|  \leq  |R_1| + |\Delta_1|$. As $D^h_l$ has
slope $-1$, we get $|\Delta_1| = \frac12 (\beta_1^h)^2$. Moreover,
we have $|R_1| = \rho_2^h \cdot \beta_1^h \leq \frac12(\rho_2^h)^2 +
\frac12 (\beta_1^h)^2$. Altogether, we get the stated bound.
\end{proof}

Thus, we charge the bottom side\footnote{The charge to the bottom of
$\SqLoch_1$ can be reduced to $\frac 34$ by considering the larger
one of the rectangles, $R_1$ and the one induced by $Q$, $Q'$, and
$P$, as well as the triangle below the larger rectangle formed by
$D_l^h$ and $D_r^h$. However, this does not lead to a better
competitive ratio, because these costs are already dominated by the
cost for holes of Type~II.} of $\SqLoch_1$ with~$1$ and the right
side of $\SqLoch_2$ with~$\frac12$. In this case, we get
$c^\beta_{h,1}=1$ and $c^\rho_{h,2}=\frac{1}{2}$.

If $\hat{H}_h$ is of Type~II, we define $R_1$ and $\Delta_1$ in the
same way. In addition, $R_2$ is the rectangle of area
$\beta_k^h\cdot \lambda_{k-1}^h$ induced by the points $Q'$ and
$P_{k-1,k}$ as well as the part of $B_{\SqLoch_k}$ that belongs to
$\partial \hat{H}_h$. Let $\Delta_2$ be the triangle below $R_2$,
induced by the bottom side of $R_2$, $D_r^h$, and $V_Q$. Using
similar arguments as in the preceding lemma, we get:

\begin{corollary}\label{cor:sizecatII}
Let $\hat{H}_h$ be a hole of Type~II.
Then,
$$|\hat{H}_h|
\leq (\beta_1^h)^2 + (\beta_k^h)^2 + \frac12 (\rho_2^h)^2 + \frac12
(\lambda_{k-1}^h)^2\; .
$$
\end{corollary}
We obtain the charges $c^\beta_{h,1}=1$, $c^\rho_{h,2}=\frac{1}{2}$,
$c^\beta_{h,k}=1$ and $c^\lambda_{h,k-1}=\frac{1}{2}$. Thus, we have
a maximum total charge of $2$ (bottom:~1, left:~1/2, and right:~1/2)
for a square, so far.

\paragraph{Holes with a Virtual Lid\label{sect:virtual}}
Next we consider a hole, $\hat{H}_h$, with a virtual lid. Let
$\hat{H}_g$ be the hole immediately above $\hat{H}_h$, {\em i.e.,}
$\hat{H}_h$ was created by removing the diagonal-boundary
intersections in $\hat{H}_g$. Corresponding to
Lemma~\ref{lem:caseAcaseB_1}, let $\SqLoch_{\mathrm{up}}$ be the
square whose copy becomes a new lid, while $\SqLoch'_{\mathrm{up}}$
is the copy. The bottom neighbor of $\SqLoch_{\mathrm{up}}$ is
denoted by $\SqLoch_{\mathrm{low}}$. We show that
$\SqLoch'_{\mathrm{up}}$ increases the total charge of
$\SqLoch_{\mathrm{up}}$ not above $2.5$. Recall that $\hat{H}_h$ is
a hole of Type~I or II by Corollary~\ref{cor:hstar}.


\begin{figure}[t]
\psfrag{Hgh}{$\hat{H}_g$} %
\psfrag{Hhh}{$\hat{H}_h$} %
\psfrag{Aupt}{$\tilde{A}_{\mathrm{up}}$} %
\psfrag{X}{$\tilde{A}_{\mathrm{low}}$} %
\psfrag{Aqt}{$\tilde{A}_q$} %
\psfrag{Aqm1t}{$\tilde{A}_{q-1}$} %
\psfrag{E}{$E$} %
\psfrag{F}{$F$} %
\psfrag{Pprime}{$P'$} %
\psfrag{P}{$P$} %
\psfrag{N}{$N$} %
\psfrag{Dlg}{$D_l^g$} %
\psfrag{VN}{$V_N$} %
\psfrag{R1}{$R_1$} %
\psfrag{r}{$\rho_{\mathrm{low}}^h$} %
\psfrag{b}{$\beta_{\mathrm{up}}^h$} %
\centering
\includegraphics[height=6.5cm]{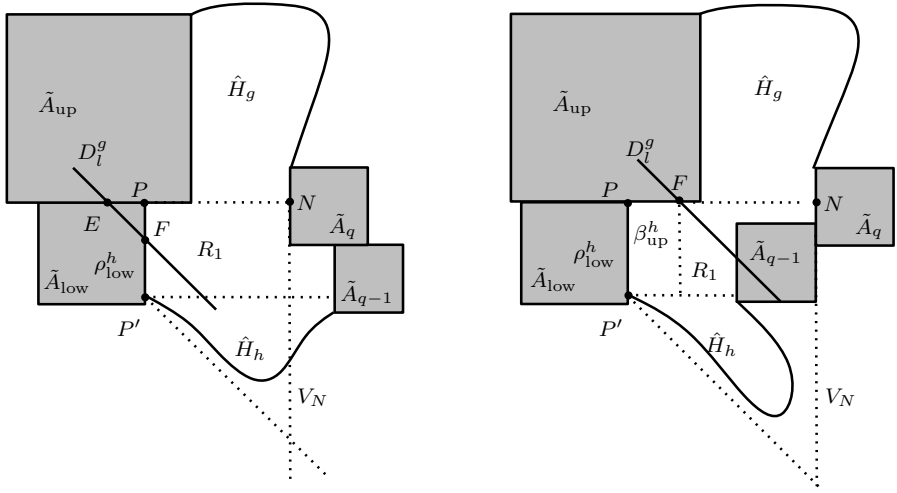}
\caption{The holes $\hat{H}_g$ and $\hat{H}_h$ and the rectangle
$R_1$ which is divided into two parts by $D_l^g$. The upper part is
already included in the bound for $\hat{H}_g$. The lower part is
charged completely to $R_{\SqLoch_\mathrm{low}}$ and
$B_{\SqLoch_\mathrm{up}'}$. Here $P$ and $P'$ are defined w.r.t.\
$\hat{H}_h$.} \label{fig:virtuallid}
\end{figure}






If $\SqLoch_{\mathrm{up}}$ does not exceed $\SqLoch_{\mathrm{low}}$
to the left, it cannot serve as a lid for any other hole; see
Fig.~\ref{fig:virtuallid}. Hence, the charge of the bottom side of
$\SqLoch_{\mathrm{up}}$ is zero; by Corollary~\ref{cor:hstar},
Lemma~\ref{cor:sizecatI}, and Corollary~\ref{cor:sizecatII} we
obtain a charge of at most $1$ to the bottom side of
$\SqLoch'_{\mathrm{up}}$. Thus, we get a total charge of $1$ to
$\SqLoch_{\mathrm{up}}$. For an easier summation of the charges at
the end, we transfer the charge from the bottom side of
$\SqLoch'_{\mathrm{up}}$ to the bottom side of
$\SqLoch_{\mathrm{up}}$.

If it exceeds $\SqLoch_{\mathrm{low}}$ to the left, we know that the
part $B_{\SqLoch_{\mathrm{up}}} \cap T_{\SqLoch_{\mathrm{low}}}$ of
$B_{\SqLoch_{\mathrm{up}}}$ is not charged by any other hole,
because it does not belong to the boundary of a hole, and the lid is
defined uniquely.

We define points, $P$ and $P'$, for $\hat{H}_h$ in the same way as
in the preceding section. Independent of $\hat{H}_h$'s type,
$\SqLoch'_{\mathrm{up}}$ would get charged only for the rectangle
$R_1$ induced by $P$, $P'$, and $N$, as well as for the triangle
below $R_1$ if we would use Lemma~\ref{cor:sizecatI} and
Corollary~\ref{cor:sizecatII}.

Next we show that we do not have to charge $\SqLoch'_{\mathrm{up}}$
for $R_1$ at all, because the part of $R_1$ that is above $D_l^g$ is
already included in the bound for $\hat{H}_g$, and the remaining
part can be charged to $B_{\SqLoch_{\mathrm{up}}}$ and
$R_{\SqLoch_{\mathrm{low}}}$. $\SqLoch'_{\mathrm{up}}$ will get
charged only $\frac 12$ for the triangle.

$D_l^g$ splits $R_1$ into a part that is above this line and a part
that is below this line. The latter part of $R_1$ is not included in
the bound for $\hat{H}_g$. Let $F$ be the intersection of $\partial
\hat{H}_g$ and $D_l^g$ that caused the creation of $\hat{H}_h$. If
$F \in R_{\SqLoch_{\mathrm{low}}}$, then this part is at most $\frac
12 (\rho_{\mathrm{low}}^h)^2$, where $\rho_{\mathrm{low}}^h$ is the
length of $\overline{P'F}$.
We charge $\frac 12$ to $R_{\SqLoch_{\mathrm{low}}}$. If $F \in
B_{\SqLoch_{\mathrm{up}}}$, then the part of $R_1$ below $D_l^g$ can
be split into a rectangular part of area $\rho_{\mathrm{low}}^h
\cdot \beta_{\mathrm{up}}^h$,
and a triangular part of area $\frac 12 (\rho_{\mathrm{low}}^h)^2$;
see Fig.~\ref{fig:virtuallid}. Here $\beta_{\mathrm{up}}^h$ is the
length of $\overline{PF}$.
The cost of the triangular part is charged to
$R_{\SqLoch_{\mathrm{low}}}$. Note that $B_{\SqLoch_\mathrm{up}}$
exceeds $\SqLoch_{\mathrm{low}}$ to the left and to the right and
that the part that exceeds $\SqLoch_{\mathrm{low}}$ to the right is
not charged. Moreover,
$\rho_{\mathrm{low}}^h$
is not larger than $B_{\SqLoch_{\mathrm{up}}} \cap
T_{\SqLoch_{\mathrm{low}}}$, {\em i.e.,} the part of
$B_{\SqLoch_{\mathrm{up}}}$ that was not charged before. Therefore,
we can charge the rectangular part completely to
$B_{\SqLoch_{\mathrm{up}}}$. Hence, $\SqLoch'_{\mathrm{up}}$ is
charged $\frac 12$ for the triangle below $R_1$, and
$\SqLoch_{\mathrm{up}}$ is charged at most $2.5$ in total.


\paragraph{Holes Containing Parts of
{\boldmath $ \partial S$}\label{sect:strip}} So far we did not
consider holes whose boundary touches $\partial S$. We show in this
section that these holes are just special cases of the ones
discussed in the preceding sections.

Because the top side of a square never gets charged for a hole, it
does not matter whether a part of $B_S$ belongs to the boundary.
Moreover, for any hole, $\hat{H}_h$, either $L_S$ or $R_S$ can be a
part of $\partial \hat{H}_h$, because otherwise there exits a curve
with one endpoint on $L_S$ and the other endpoint on $R_S$, with the
property that this curve lies completely inside of $\hat{H}_h$. This
contradicts the existence of the bottom sequence of a square lying
above the curve.

For a hole $\hat{H}_h$ with $L_S$ contributing to $\partial
\hat{H}_h$, we can use the same arguments as in the proof for
Lemma~\ref{lem:directionoftraversal} to show that $L_S \cap \partial
\hat{H}_h$ is a single line segment. Let $P$ be the topmost point of
this line segment and $\SqLoch_1$ be the square containing $P$.
$\SqLoch_1$ must have a bottom neighbor, $\SqLoch_k$, and
$\SqLoch_k$ must have a left neighbor, $\SqLoch_{k-1}$, we get
$P_{k,1}\in B_{\SqLoch_1}$ and $P_{k-1,k} \in L_{\SqLoch_k}$,
respectively. We define the right diagonal, $D_r$, and the point
$Q'$ as above and conclude that $\hat{H}_h$ lies completely inside
the polygon formed by $L_S \cap
\partial \hat{H}_h$, $D_r$, and the part of $\partial \hat{H}_h$ that
is between $P$ and $Q'$ in clockwise order. We split this polygon
into a rectangle and a triangle in order to obtain charges of $1$ to
$B_{\SqLoch_1}$ and $\frac{1}{2}$ to $L_{\SqLoch_k}$.

Now, consider a hole where a part of $R_S$ belongs to $\partial
\hat{H}_h$. We denote the topmost point on $R_S \cap \partial
\hat{H}_h$ by $Q$, and the square containing $Q$ by $\SqLoch_1$.
We number the squares in counterclockwise order and define the left
diagonal, $D_l$, as above. Now we consider the intersections of
$D_l$ and eliminate them by creating new holes. After this, the
modified hole $\hat{H}_h^{(z)}$ can be viewed as a hole of Type~II,
for which the part on the right side of $V_Q$ has been cut off;
compare Corollary~\ref{cor:sizecatII}. We obtain charges of $1$ to
$B_{\SqLoch_1}$ and $\frac{1}{2}$ to $R_{\SqLoch_2}$. For the copy
of a square we get a charge of $\frac{1}{2}$ to the bottom side.

\subsection{Summing up the Charges}\label{sect:summingupcharges}

Altogether, we have the charges from Table~\ref{tab:charges}. The
charges depend on the type of the adjacent hole (Type~I,~II,
touching or not touching the strip's boundary), but the maximal
charge dominates the other ones. Moreover, the square may also serve
as a virtual lid. The maximal charges from a hole with non-virtual
lid and those from a hole with virtual lid sum up to a total charge
of $2.5$ per square. This proves our claim from the beginning:
\begin{equation*}
\sum_{h=1}^s |H_h| \leq 2.5 \cdot \sum_{i=1}^{n+1} a_{i}^2 \, .
\end{equation*}


\begin{table}[t]
\scriptsize
\begin{center}
\begin{tabular}{|l||*5{c|}|*4{c|}|c||} \hline
& \multicolumn{5}{c||}{Non-virtual Lid} &
\multicolumn{4}{c||}{Virtual Lid} & Total
\\ \hline
& Type~I & Type~II & $L_S$ & $R_S$ & Max.\ & Type~I & Type~II &
$R_S$ & Max.\ &
\\ \hline
\rule[-7pt]{0pt}{20pt}Left Side    &         0 & $\frac12$ & $\frac12$ &         0 & $\frac12$  &         0 &         0 &         0 &         0 & $\frac12$ \\ \hline %
\rule[-7pt]{0pt}{20pt}Bottom Side  &         1 &         1 &         1 &         1 & 1          & $\frac12$ & $\frac12$ & $\frac12$ & $\frac12$ & 1.5       \\ \hline %
\rule[-7pt]{0pt}{20pt}Right Side   & $\frac12$ & $\frac12$ &         0 & $\frac12$ & $\frac12$  &         0 &         0 &         0 &         0 & $\frac12$ \\ \hline\hline %
\rule[-7pt]{0pt}{20pt}Total &&&&& 2 &&&& $\frac12$ & 2.5\\ \hline
\end{tabular}
\end{center}
\caption{The charges to the different sides of a single square.
Summing up the charges to the different sides, we conclude that
every square gets a total charge of at most 2.5.
\label{tab:charges}}
\end{table}

\section{The Strategy {\em SlotAlgorithm}}\label{sect:stratSlotAlg}

In this section we analyze a different strategy for the strip
packing problem with Tetris and gravity constraint. This strategy
provides more structure on the generated packing, which allows us to
prove an upper bound of $2.6154$ on the asymptotic competitive
ratio.

\subsection{The Algorithm}

Consider two vertical lines of infinite length going upwards from
the bottom side of $S$ and parallel to the left and the right side
of $S$. We call the area between these lines a {\em slot}, the lines
the {\em left boundary} and the {\em right boundary} of the slot,
and the distance between the lines the {\em width} of the slot.

Now, our strategy {\em SlotAlgorithm}\index{SlotAlgorithm} works as
follows: We divide the strip $S$ of width $1$ into slots of
different widths; for every $j=0,1,2\ldots$, we create $2^j$ slots
of width $2^{-j}$ side by side, {\em i.e.,} we divide $S$ into one
slot of width 1, two slots of width $1/2$, four slots of width
$1/4$, and so on. Note that a slot of width $2^{-j}$ contains 2
slots of width $2^{-j-1}$; see Fig.~\ref{fig:shadows}.

For every square $A_i$, we round the side length $a_i$ to the
smallest number $2^{-k_i}$ that is larger than or equal to $a_i$.
Among all slots of with $2^{-k_i}$, we place $A_i$ in the one that
allows $A_i$ to be placed as near to the bottom of $S$ as possible,
by moving $A_i$ down along the left boundary of the chosen slot
until another square is reached.
The algorithm clearly satisfies the Tetris and the gravity
constraints, and next we show that the produced height is at most
2.6154 times the height of an optimal packing.

\begin{figure}[t]
\psfrag{A1}{$A_1$} %
\psfrag{A1S}{$A_1^S$} %
\psfrag{A2}{$A_2$} %
\psfrag{A2S}{$A_2^S$} %
\psfrag{A2W}{$A_2^W$} %
\psfrag{A3}{$A_3$} %
\psfrag{A3S}{$A_3^S$} %
\psfrag{A3W}{$A_3^W$} %
\psfrag{T2}{$T_2$} %
\psfrag{T3}{$T_3$} %
\psfrag{Q}{$Q$} %
\psfrag{R}{$R$} %
\psfrag{P}{$P$} %
\psfrag{delta2}{$\delta_2$} %
\psfrag{delta2prime}{$\delta_2'$} %
\psfrag{delta3prime}{$\delta_3'$} %
\psfrag{a3mdelta3prime}{$a_3\!-\!\delta_3'$} %
\psfrag{2hk1m2}{$2^{-k_1+2}$} %
\psfrag{2hk1m1}{$2^{-k_1+1}$} %
\psfrag{2hk1}{$2^{-k_1}$} %
\centering
\includegraphics[height=6.5cm]{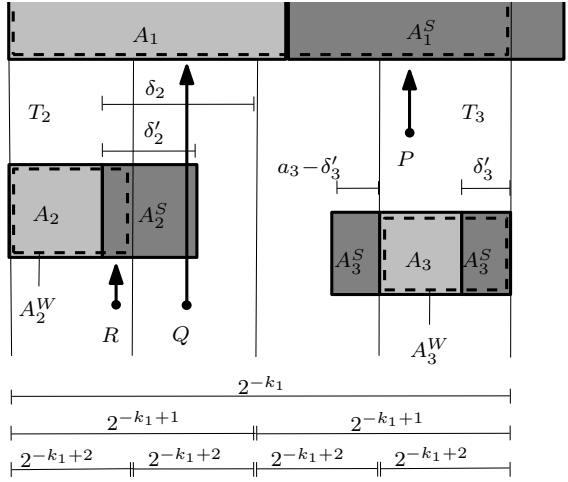}
\caption{Squares $A_i$, $i=1,2,3$, with their {\em shadows} $A_i^S$
and their {\em widening} $A_i^W$. $\delta_2'$ is equal to $a_2$ and
$\delta_3'$ is equal to $\delta_3$. The points $P$ and $Q$ are
charged to $A_1$. $R$ is not charged to $A_1$, but to $A_2$.}
\label{fig:shadows}
\end{figure}

\subsection{Analysis}

Let $A_i$ be a square placed by the SlotAlgorithm in a slot $T_i$ of
width $2^{-k_i}$. If $a_i\leq \frac 12$, we define $\delta_i$ as the
distance between the right side of $A_i$ and the right boundary of
the slot of width $2^{-k_i+1}$ that contains $A_i$, and we define
$\delta_i'=\min\{a_i, \delta_i\}$. We call the area obtained by
enlarging $A_i$ by $\delta_i'$ to the right and by $a_i - \delta_i'$
to the left (without $A_i$ itself)
the {\em shadow}\index{shadow} of $A_{i}$ and denote it by
$A_{i}^S$. Thus, $A_{i}^S$ is an area of the same size as $A_{i}$
and lies completely inside a slot of twice the width of $A_{i}$'s
slot. If $a_i\geq \frac 12$, we enlarge $A_i$ only to the right side
and call this area the shadow. Moreover, we define the {\em
widening}\index{widening} of $A_i$ as $A_{i}^W = (A_{i} \cup A_{i}^S
) \cap T_i$; see Fig.~\ref{fig:shadows}.

Now, consider a point $P$ in $S$ that is not inside an $A_{j}^W$ for
any square $A_{j}$.
We charge $P$ to the square, $A_{i}$, if $A^W_i$ is the first
widening that intersects the vertical line going upwards from $P$.
We denote by $F_{A_{i}}$ the set of all points charged to $A_{i}$
and by $|F_{A_{i}}|$ its area. For points lying on the left or the
right boundary of a slot, we break ties arbitrarily. For the
analysis, we place a closing square, $A_{n+1}$, of side length~1 on
top of the packing. Therefore, every point in the packing that does
not lie inside an $A_{j}^W$
is charged to a square. Because $A_i$ and $A_i^S$ have the same
area, we can bound the height 
of the packing produced by the SlotAlgorithm by
\begin{equation*}
2 \cdot \sum_{i=1}^{n} a_i^2 + \sum_{i=1}^{n+1} |F_{A_i}| \;.
\end{equation*}

\begin{theorem}\label{fks-osp-main-theo}
The SlotAlgorithm is (asymptotically) 2.6154-competitive.
\end{theorem}

\begin{proof}
The height of an optimal packing is at least $\sum_{i=1}^{n} a_i^2$,
and therefore, it suffices to show that $|F_{A_i}| \leq 0.6154 \cdot
a_i^2$ holds for every square $A_i$. We construct for every $A_i$ a
sequence of squares $\SqLoch_1^i, \SqLoch_2^i, \ldots, \SqLoch_m^i$
with $\SqLoch_1^i= A_i$ (to ease notation, we omit the superscript
$i$ in the following). We denote by $E_{j}$ the extension of the
bottom side of $\SqLoch_j$ to the left and to the right; see
Fig.~\ref{fig:startinduction}.

We will show that by an appropriate choice of the sequence, we can
bound the area of the part of $F_{\SqLoch_1}$ that lies between a
consecutive pair of extensions, $E_j$ and $E_{j+1}$, in terms of
$\SqLoch_{j+1}$ and the slot width. From this we will derive the
upper bound on the area of $F_{\SqLoch_1}$. We assume throughout the
proof that the square $\SqLoch_j$, $j\geq 1$, is placed in a slot,
$T_j$, of width $2^{-k_j}$. Note that $F_{\SqLoch_1}$ is completely
contained in $T_1$.

A slot is called {\em active (with respect to $E_j$ and
$\SqLoch_1$)} if there is a point in the slot that lies below $E_j$
and that is charged to $\SqLoch_1$ and {\em nonactive} otherwise. If
it is clear from the context we leave out the $\SqLoch_1$.

The sequence of squares is chosen as follows: $\SqLoch_1$ is the
first square and the next square, $\SqLoch_{j+1}$, $j=1,\ldots,m-1$,
is chosen as the smallest one that intersects or touches $E_j$ in an
active slot (w.r.t.\ $E_j$ and $\SqLoch_1$) of width $2^{-k_j}$ and
that is not equal to $\SqLoch_j$; see Fig.~\ref{fig:startinduction}.
The sequence ends if all slots are nonactive w.r.t.\ to
an extension $E_m$. We prove each of the following claims by induction: \\

\noindent{\bf Claim A:} $\SqLoch_{j+1}$ exists for $j+1 \leq m$ and
$\SqLochklein_{j+1} \leq 2^{-k_j-1}$ for $j+1\leq m$, {\em i.e.,}
the sequence exists and its elements have decreasing side length. \\

\noindent{\bf Claim B:}   The number of active slots (w.r.t.~$E_j$)
of width $2^{-k_j}$ is at most
\begin{center}
$
\begin{array}{cl}
1 & \mbox{, for } j=1 \mbox{ and } \\
\prod_{i=2}^j (\frac{1}{2^{k_{i-1}}} 2^{k_i} -1 ) &  \mbox{, for } j
\geq 2 \; .
\end{array}
$
\end{center}


\noindent{\bf Claim C:} The area of the part of $F_{\SqLoch_1}$ that
lies in an active slot of width $2^{-k_j}$ between $E_j$ and
$E_{j+1}$ is at most ${2^{-k_j}}
\SqLochklein_{j+1} - 2 \SqLochklein_{j+1}^2$. \\

\begin{figure}
\psfrag{A1t}{$\tilde{A}_1$} %
\psfrag{A2t}{$\tilde{A}_2$} %
\psfrag{A3t}{$\tilde{A}_3$} %
\psfrag{T2}{$T_2$} %
\psfrag{T3}{$T_3$} %
\psfrag{E1}{$E_1$} %
\psfrag{E2}{$E_2$} %
\psfrag{E3}{$E_3$} %
\psfrag{X}{$\leq 2^{-k_2} \cdot \tilde{a}_3 - 2 \tilde{a}_3^2$}
\psfrag{2hmk1}{$2^{-k_1}$} %
\psfrag{2hmk2}{$2^{-k_2}$} %
\psfrag{2hmk3}{$2^{-k_3}$} %
\centering
\includegraphics[height=6.3cm]{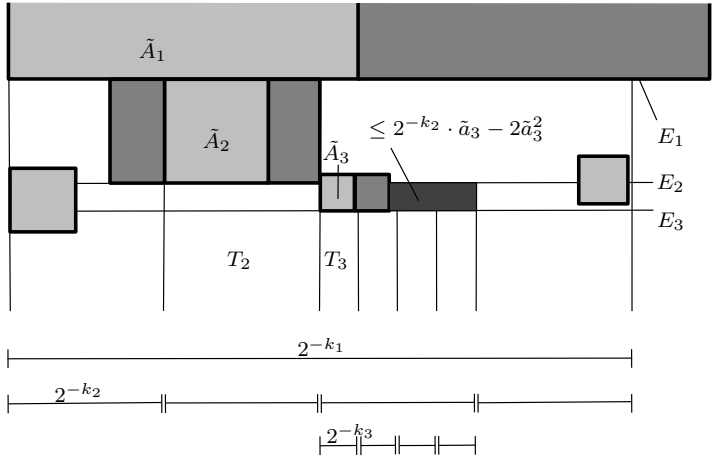} 
\caption{The first three squares of the sequence (light gray) with
their shadows (gray). In this example, $\SqLoch_2$ is the smallest
square that bounds $\SqLoch_1$ from below. $\SqLoch_3$ is the
smallest one that intersects $E_2$ in an active slot (w.r.t.\ $E_2$)
of width $2^{-k_2}$. There has to be an intersection of $E_2$ and
some square in every active slot because, otherwise, there would
have been a better position for $\SqLoch_2$. $T_2$ is nonactive,
(w.r.t.\ $E_2$) and of course, also w.r.t.\ all extension $E_j$,
$j\geq 3$. The part of $F_{\SqLoch_1}$ that lies between $E_1$ and
$E_2$ has size $2^{-k_1} \SqLochklein_{2} - 2 \SqLochklein_{2}^2$.}
\label{fig:startinduction}
\end{figure}

\noindent{\bf Proof of Claim A:} If $\SqLoch_{1}$ is placed on the
bottom of $S$, $F_{\SqLoch_1}$ has size 0 and $\SqLoch_1$ is the
last element of the sequence. Otherwise, the square $\SqLoch_1$ has
at least one bottom neighbor, which is a candidate for the choice of
$\SqLoch_2$.

Now suppose for a contradiction that there is no candidate for the
choice of the $(j+1)$th element.
Let $T'$ be an active slot in $T_1$ (w.r.t.\ $E_j$) of width
$2^{-k_j}$ where $E_j$ is not intersected by a square in $T'$. If
there is an $\varepsilon$ such that for every point, $P\in ( T' \cap
E_j )$, there is a point, $P'$, at a distance $\varepsilon$ below
$P$ which is charged to $\SqLoch_1$, we conclude that there would
have been a better position for $\SqLoch_{j}$. Hence, there is at
least one point, $Q$, below $E_j$ that is not charged to
$\SqLoch_1$; see Fig~\ref{fig:existence}. Consider the bottom
sequence (as defined in Section~\ref{sect:basicproperties}) of the
square $Q$ is charged to. This sequence must intersect $E_j$ outside
of $T'$ (by the choice of $T'$). This implies that one of its
elements must intersect the left or the right boundary of $T'$, and
we can conclude that this square has at least the width of $T'$.
This is because (by the algorithm) a square with rounded side length
$2^{-\ell}$ cannot cross a slot's boundary of width larger than
$2^{-\ell}$. In turn, a square with rounded side length larger than
the width of $T'$ completely covers $T'$, and $T'$ cannot be active
w.r.t.\ to $E_j$ and $\SqLoch_1$. Thus, all points in $T'$ below
$E_j$ are charged to this square; a contradiction. This proves that
there is a candidate for the choice of $\SqLoch_{j+1}$.

\smallskip

Suppose $\SqLochklein_{2} > 2^{-k_1-1}$. Then, $\SqLoch_2$ was
placed in a slot of width at least $2^{-k_1}$. Thus, its widening
has width at least $2^{-k_1}$, and $\SqLoch_2$ is a bottom neighbor
of $\SqLoch_1$. Then, no point in $T_1$, below $E_1$, is charged to
$\SqLoch_1$, and hence, $T_1$ is nonactive w.r.t.\ $E_1$ and
$\SqLoch_1$. This implies, that $\SqLoch_2$ does not belong to the
sequence; a contradiction.


Because we chose $\SqLoch_{j+1}$ to be of minimal side length,
$\SqLochklein_{j+1} \geq 2^{-k_j}$ would imply that all slots inside
$T$ are nonactive (w.r.t.\ $E_j$). Therefore, if $\SqLoch_{j+1}$
belongs to the sequence,
$\SqLochklein_{j+1} \leq 2^{-k_j-1}$ must hold. \\

\begin{figure}
\psfrag{A1t}{$\tilde{A_1}$} %
\psfrag{Ah}{$\hat{A}$} %
\psfrag{Ajt}{$\tilde{A_j}$} %
\psfrag{2hmkj}{$2^{-k_j}$} %
\psfrag{Q}{$Q$} %
\psfrag{Tprime}{$T'$} %
\psfrag{bs}{${\cal{B}}_{\hat{A}}$} 
\psfrag{eps}{$\varepsilon$}
\psfrag{Ej}{$E_j$}
\centering
\includegraphics[height=6.4cm]{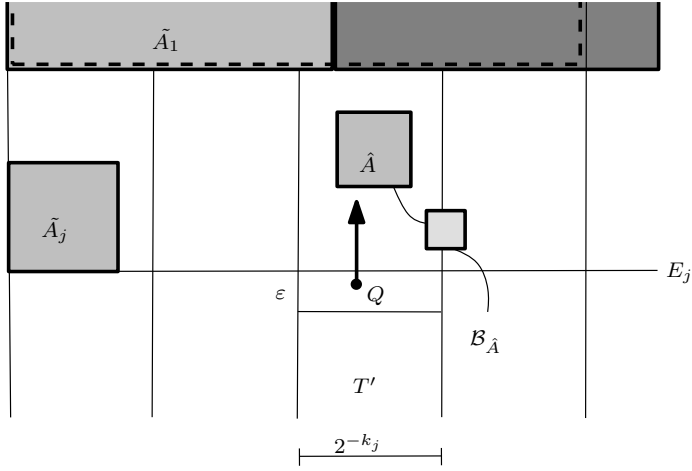}
\caption{If $E_j$ is not intersected in an active slot of size
$2^{-k_j}$ we obtain a contradiction: Either there is a position for
$A_j$ that is closer to the bottom of $S$ or there is a square that
makes $E_j$ nonactive. $\hat{A}$ is the square $Q$ is charged to,
${\cal{B}}_{\hat{A}}$ its bottom sequence.
} \label{fig:existence}
\end{figure}

\noindent{\bf Proof of Claim B:} Obviously, there is at most one
active slot of width $2^{-k_1}$; see Fig.~\ref{fig:startinduction}.
By the induction hypothesis, there are at most
\begin{equation*}
(\frac{1}{2^{k_1}} 2^{k_2} -1 ) \cdot (\frac{1}{2^{k_2}} 2^{k_3} -1
)\cdot \ldots \cdot(\frac{1}{2^{k_{j-2}}} 2^{k_{j-1}} -1 )
\end{equation*}
active slots of width $2^{-k_{j-1}}$ (w.r.t.\ $E_{j-1}$). Each of
these slots contains $2^{k_{j}-k_{j-1}}$ slots of width $2^{-k_j}$,
and in every active slot of width $2^{-k_{j-1}}$ at least one slot
of width $2^{-k_j}$ is nonactive because we chose $\SqLoch_j$ to be
of minimum side length. Hence, the number of active slots (w.r.t.\
$E_j$) is a factor of $(\frac{1}{2^{k_{j-1}}} 2^{k_{j}} -1 )$ times
larger than the number of active slots (w.r.t.\
$E_{j-1}$).  \\

\noindent{\bf Proof of Claim C:} The area of the part of
$F_{\SqLoch_1}$ that lies between $E_1$ and $E_2$ is at most
$2^{-k_1} \SqLochklein_{2} - 2 \SqLochklein_{2}^2$; see
Fig.~\ref{fig:startinduction}. Note that we can subtract the area of
$\SqLoch_2$ twice, because $\SqLoch_2^S$ was defined to lie
completely inside a slot of width $2^{-k_2+1}\leq 2^{-k_1}$ and is
of same area as $\SqLoch_2$.

By the choice of $\SqLoch_{j+1}$ and because in every active slot of
width $2^{-k_j}$ there is at least one square that intersects $E_j$
(points below the widening of this square are not charged to
$\SqLoch_1$) we conclude that the area of $F_{\SqLoch_1}$ between
$E_j$ and $E_{j+1}$ is at most $2^{-k_j} \SqLochklein_{j+1} - 2
\SqLochklein_{j+1}^2$, in every active slot of width $2^{-k_j}$.
\\

{\bf Altogether}, we proved that the sequence is well defined and we
calculated an upper bound on the number of active slots and an upper
bound on the size of the part of $|F_{\SqLoch_1}|$ that is contained
in an active slot. Multiplying the number and the size yields an
upper bound on  $|F_{\SqLoch_1}|$ of
\begin{equation*}
|F_{\SqLoch_1}| \leq (\frac{\SqLochklein_{2}}{2^{k_1}}  -
2\SqLochklein_{2}^{2} ) \cdot 1 + \sum\limits_{j=2}^m \!\left(
\frac{\SqLochklein_{j+1}}{2^{k_j}} \!-\! 2
\SqLochklein_{j+1}^{2}\right) \prod\limits_{i=2}^{j}
\left(\frac{2^{k_{i}}}{2^{k_{i-1}}}\!-\!1 \right) \; .
\end{equation*}
This expression is maximized if we choose
$\SqLochklein_{i+1}=1/2^{k_i +2}$, for $i=1,\ldots,m$, {\em i.e.,}
$k_i=k_1+2(i-1)$.

\noindent We get:
\begin{eqnarray*}
|F_{\SqLoch_1}| & \leq & \frac{1}{2^{k_1+2}} \cdot \frac{1}{2^{k_1}}
- 2\cdot \left ( \frac{1}{2^{k_1+2}}\right )^2 \\
& & + \sum_{j=2}^{m} \left [ \frac{1}{2^{k_1+2(j-1)}} \cdot
\frac{1}{2^{k_1+2j}} - 2\left (\frac{1}{2^{k_1+2j}} \right )^2
\right ] \prod_{i=1}^{j-1} \left ( \frac{2^{k_1+2i}}{2^{k_1+2(i-1)}}
-1 \right )  \\
& = & \frac{1}{2^{k_1+3}} + \sum_{j=2}^{m} \left [
\frac{1}{2^{2k_1+4j-2}} - \frac{1}{2^{2k_1+4j-1}}\right ] \cdot
3^{j-1} \\
& = & \frac{1}{2^{k_1+3}} + \sum_{j=2}^{m}
\frac{3^{j-1}}{2^{2k_1+4j-1}} \\
& = & \frac{1}{2^{k_1+3}} + \sum_{j=1}^{m-1} \frac{3^{j}}{2^{2k_1+4j
+3}} \\
& \leq & \sum_{j=0}^{\infty} \frac{3^{j}}{2^{2k_1+4j+3}} \; .
\end{eqnarray*}
The fraction $|F_{\SqLoch_1}|/\SqLochklein_1^2$ is maximized, if we
choose $\SqLochklein_{1}$ as small as possible, {\em i.e.,}
$\SqLochklein_{1}=2^{-k_1-1}+\varepsilon$. We conclude:
\begin{equation*}
\frac{|F_{\SqLoch_1}|}{\SqLochklein_{1}^2} \leq \sum_{j=0}^{\infty}
\frac{2^{2k_1+2} \cdot 3^j}{2^{2k_1+4j+3}} = \sum_{j=0}^{\infty}
\frac{3^j}{2^{4j+1}} = \frac{1}{2} \cdot \sum_{j=0}^{\infty}
\left(\frac{3}{16}\right)^j = \frac{8}{13} = 0.6153...
\end{equation*}
Thus,
\begin{equation*}
|F_{\SqOrig_i}|
\leq 0.6154 \cdot a_{i}^2 \; .
\end{equation*}
\end{proof}

\section{Lower Bounds}

The lower bound construction for online strip packing introduced by
Galambos and Frenk~\cite{gf-spllb-93} and later improved by van
Vliet~\cite{v-ilbob-92} relies on an integer programming formulation
and its LP-relaxation for a specific bin packing instance. This
formulation does not take into account that there has to be a
collision free path to the final position of the item. Hence, it
does not carry over to our setting.

\begin{figure}
\psfrag{54}{$\frac 54$} %
\psfrag{1}{$1$} %
\psfrag{14}{$\frac 14$} %
\psfrag{12}{$\frac 12$} %
\psfrag{12e}{$\frac 12+\varepsilon$} %
\psfrag{34e}{$\frac 34+\varepsilon$} %
\psfrag{alg}{any algorithm} %
\psfrag{theopt}{the optimum} %
\psfrag{t1}{Type I} %
\psfrag{t2}{Type II} %
\centering
\includegraphics[height=7.6cm]{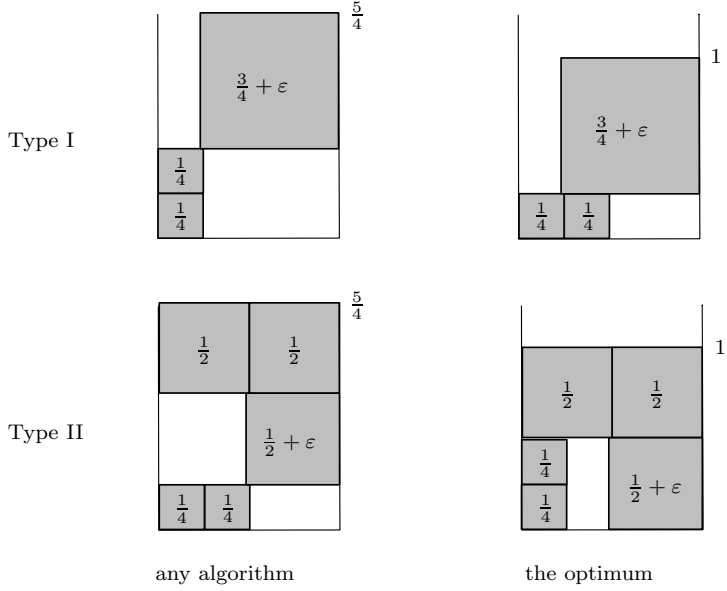}
\caption{The left column shows possible packings of any algorithm
for one iteration. The right column contains optimal packings. The
top row displays the first and the bottom row the second type of
sequence. } \label{fig:lowerbound}
\end{figure}

The best asymptotic lower bound, we are aware of, is $\frac 54$. It
is based on two sequences which are repeated iteratively. We denote
by $A_i^k$, $k=1,\ldots,5$ and $i=1,2,\ldots$, the $k$-th square of
the sequence in the $i$-th iteration, and we denote by $H_i$,
$i=1,2,\ldots$, the height of the packing after the $i$-th
iteration; we define $H_0=0$.

The first two squares of each sequence have a side length of $\frac
14$, that is, $a_i^1=a_i^2=\frac14$.
Now, depending on the choice of the algorithm, the sequence
continues with one of the following two possibilities (see
Fig.~\ref{fig:lowerbound}):

Type~I: If the algorithms packs the first two squares on top of each
other, with the bottom side of the lower square at height $H_{i-1}$,
the sequence continues with a square of side length $\frac 34 +
\varepsilon$, {\em i.e.,} $a_i^3= \frac 34 + \varepsilon$ and
$a_i^4=a_i^5=0$ (upper left picture in Fig.~\ref{fig:lowerbound}).

Type~II: Otherwise,
the sequence continues with a square of side length $\frac 12 +
\varepsilon$ and two squares of side length $\frac 12$, {\em i.e.,}
$a_i^3= \frac 12 + \varepsilon$ and $a_i^4=a_i^5=\frac 12$ (lower
left picture in Fig.~\ref{fig:lowerbound}).

\begin{lemma}\label{lem:lowerbound}
The height of the packing produced by any algorithm increases in
each iteration, on average, by at least $\frac 54$.
\end{lemma}
\begin{proof}
Consider the $i$-th iteration, $i=1,2,\ldots$. If the sequence is of
Type~I, the statement is obviously true because the square of side
length $\frac 34 + \varepsilon$ cannot pass any of the squares of
side length $\frac 14$ which are packed on top of each other; see
Fig.~\ref{fig:lowerbound}.

If the sequence is of Type~II, we need to consider the previous
iteration. If there was no previous iteration, then we know that
$A_i^1$ and $A_i^2$ are both placed on the bottom side of the strip.
Because $A_i^3$ cannot be placed on the bottom side, and $A_i^4$ and
$A_i^5$ cannot pass $A_i^3$, we get an increase of at least $\frac
54$.

If the sequence in the previous iteration was of Type~I, $H_{i-1}$
is determined by the square of side length $\frac 34 + \varepsilon$.
Hence, $A_i^1$ and $A_i^2$ are both placed on top of this square and
the same arguments hold.

If the sequence in the previous iteration was of Type~II, then
either $A_{i-1}^4$ and $A_{i-1}^5$ are packed next to each other or
on top of each other. In the first case, we can use the same
arguments as in the case where there was no previous iteration. In
the second case, $A_{i-1}^4$ and $A_{i-1}^5$ are placed an top of
each other {\em and} on top of $A_{i-1}^3$, because they cannot
%
pass a square with side length $\frac 12 + \varepsilon$. This
implies that, the last iteration contributed a height of at least
$\frac 32$ to the height of the packing. No matter how the algorithm
packs the squares from the current iteration (the first two squares
might be placed at the same height or even deeper as the previous
squares) it contributes a height of at least $1$ to the packing.
This proves an average increase of $\frac 54$ for both iterations.
\end{proof}

\begin{theorem}
There is no algorithm with asymptotic competitive ratio smaller than
$\frac 54$ for the online strip packing problem with Tetris and
gravity constraint.
\end{theorem}
\begin{proof}
The height of the packing produced by any algorithm increases by
$\frac 54$ per iteration for the above instance
(Lemma~\ref{lem:lowerbound}). The optimum can pack the squares
belonging to one iteration always such that the height of the
packing increases by at most $1$; see the right column of
Fig.~\ref{fig:lowerbound}.
\end{proof}

\section{Conclusion}
There are instances consisting only of squares for which the
algorithm of Azar and Epstein does not undercut its proven
competitive factor of $4$. Hence, this algorithm is tightly
analyzed. We proved competitive ratios of $3.5$ and $2.6154$ for
BottomLeft and the SlotAlgorithm, respectively. Hence, both
algorithms outperform the one by Azar and Epstein if the input
consists only of squares.

We do not know any instance for which BottomLeft
produces a packing that is $3.5$ times higher than an optimal
packing. The best lower bound we know is $\frac 54$.

Moreover, we are not aware of an instance in which the SlotAlgorithm
reaches its upper bound of $2.6154$. The instance consisting of
squares with side length $2^{-k} + \delta$, for large $k$ and small
$\delta$, gives a lower bound of $2$ on the competitive ratio.

Hence, there is still room for improvement: Our analysis might be
improved or there may be more sophisticated algorithms for the
strip packing problem with Tetris and gravity constraint.


At this point, the bottleneck in our analysis for BottomLeft is the
case in which a square has large holes at the right, left, and bottom
side and also serves as a virtual lid; see
Fig.~\ref{fig:example_bn}. This worst case can happen to only a few
squares, but never to all of them. Thus, it might be possible to
transfer charges between squares, which may yield a refined
analysis. The same holds for the SlotAlgorithm and the sequence we
constructed to calculate the size of the unoccupied area below a
square.

In addition, it may be possible to apply better lower bounds on the
packing than just the total area, {\em e.g.,} the one arising from
{\em dual-feasible functions} by Fekete and
Schepers~\cite{fs-nclbb-98}.

\bibliographystyle{plain}

\bibliography{lit}

\end{document}